\documentclass[11pt,notitlepage,tightenlines,nofootinbib,superscriptaddress,aps,pra]{revtex4-2}
\pdfoutput=1

\usepackage{newpxtext,newpxmath}

\usepackage[utf8]{inputenc}
\usepackage{newunicodechar}
\newunicodechar{ρ}{\rho}
\newunicodechar{β}{\beta}
\newunicodechar{σ}{\sigma}
\newunicodechar{λ}{\lambda}
\newunicodechar{Λ}{\Lambda}
\newunicodechar{μ}{\mu}
\newunicodechar{ν}{\nu}
\newunicodechar{ψ}{\psi}
\newunicodechar{ϕ}{\varphi}
\newunicodechar{Φ}{\Phi}
\newunicodechar{φ}{\phi}
\newunicodechar{π}{\pi}
\newunicodechar{α}{\alpha}
\newunicodechar{ϵ}{\epsilon}
\newunicodechar{ε}{\varepsilon}
\newunicodechar{δ}{\delta}
\newunicodechar{ω}{\omega}
\newunicodechar{Δ}{\Delta}
\newunicodechar{Σ}{\Sigma}

\usepackage[thmtools-compat]{keytheorems}

\usepackage{amssymb}
\usepackage{amsmath}
\usepackage{bm,bbm}
\usepackage{braket}
\usepackage{dsfont}
\usepackage{mathdots}
\usepackage{mathtools}
\usepackage{enumerate}
\usepackage[shortlabels]{enumitem}
\usepackage{csquotes}
\usepackage{stmaryrd}
\usepackage{graphicx}
\usepackage{stackengine}
\usepackage{scalerel}
\usepackage{tensor}       %\tensor[_n]{\braket{j|\psi}}{_{1\ldots n}}
\usepackage[dvipsnames,svgnames]{xcolor}
\usepackage{xfrac}
\usepackage{centernot}
\usepackage{caption}
\usepackage{comment}
\usepackage{chngcntr}
\usepackage{booktabs}
\usepackage{tabularray}
\usepackage[normalem]{ulem}
\UseTblrLibrary{booktabs}

\usepackage{caption}
\usepackage[pdftex]{hyperref}
\hypersetup{
    bookmarksnumbered=true, % If Acrobat bookmarks are requested, include section numbers
    breaklinks=true,
    unicode=false, % non-Latin characters in Acrobat bookmarks
    pdfstartview={FitH}, % fits the width of the page to the window
    pdfnewwindow=true, % links in new window
    colorlinks=true, % false: boxed links; true: colored links
    allcolors=MidnightBlue!70!black!70!TealBlue
}
\usepackage{zref-clever}
\zcsetup{cap}

\newtheorem{thm}{Theorem}
\newtheorem*{thm*}{Theorem}
\newtheorem{theorem}[thm]{Theorem}
\newtheorem*{theorem*}{Theorem}

\newtheorem*{prop*}{Proposition}
\newtheorem{proposition}[thm]{Proposition}
\newtheorem*{proposition*}{Proposition}
\newtheorem{lemma}{Lemma}
\newtheorem*{lemma*}{Lemma}

\newtheorem{corollary}{Corollary}
\newtheorem*{cor*}{Corollary}

\newtheorem*{cj*}{Conjecture}

\newtheorem*{Def*}{Definition}

\newtheorem*{question*}{Question}

\newtheorem*{problem*}{Problem}

\makeatletter
\def\thmhead@plain#1#2#3{%
  \thmname{#1}\thmnumber{\@ifnotempty{#1}{ }\@upn{#2}}%
  \thmnote{ {\the\thm@notefont#3}}}
\let\thmhead\thmhead@plain
\makeatother

\theoremstyle{definition}

\newtheorem*{rem*}{Remark}
\newtheorem*{remark}{Remark}

\newcommand{\bb}{\begin{equation}\begin{aligned}\hspace{0pt}}
\newcommand{\bbb}{\begin{equation*}\begin{aligned}}
\newcommand{\ee}{\end{aligned}\end{equation}}
\newcommand{\eee}{\end{aligned}\end{equation*}}

\newcommand\ceil[1]{\left\lceil#1\right\rceil}

\newcommand{\ketbra}[1]{\ket{#1}\!\bra{#1}}

\renewcommand{\epsilon}{\varepsilon}

\newcommand{\ve}{\varepsilon}

\DeclareMathOperator{\Tr}{Tr}

\DeclareMathAlphabet{\pazocal}{OMS}{zplm}{m}{n}

\DeclareMathOperator{\diag}{diag}

\newcommand{\lsmatrix}{\left(\begin{smallmatrix}}
\newcommand{\rsmatrix}{\end{smallmatrix}\right)}

\stackMath

\stackMath

\makeatletter
\newcommand*\rel@kern[1]{\kern#1\dimexpr\macc@kerna}
\newcommand*\widebar[1]{%
  \begingroup
  \def\mathaccent##1##2{%
    \rel@kern{0.8}%
    \overline{\rel@kern{-0.8}\macc@nucleus\rel@kern{0.2}}%
    \rel@kern{-0.2}%
  }%
  \macc@depth\@ne
  \let\math@bgroup\@empty \let\math@egroup\macc@set@skewchar
  \mathsurround\z@ \frozen@everymath{\mathgroup\macc@group\relax}%
  \macc@set@skewchar\relax
  \let\mathaccentV\macc@nested@a
  \macc@nested@a\relax111{#1}%
  \endgroup
}

\counterwithin*{equation}{part}
\counterwithin*{thm}{part}
\counterwithin*{figure}{part}

\definecolor{Blues5seq1}{RGB}{239,243,255}
\definecolor{Blues5seq2}{RGB}{189,215,231}
\definecolor{Blues5seq3}{RGB}{107,174,214}
\definecolor{Blues5seq4}{RGB}{49,130,189}
\definecolor{Blues5seq5}{RGB}{8,81,156}

\definecolor{Greens5seq1}{RGB}{237,248,233}
\definecolor{Greens5seq2}{RGB}{186,228,179}
\definecolor{Greens5seq3}{RGB}{116,196,118}
\definecolor{Greens5seq4}{RGB}{49,163,84}
\definecolor{Greens5seq5}{RGB}{0,109,44}

\definecolor{Reds5seq1}{RGB}{254,229,217}
\definecolor{Reds5seq2}{RGB}{252,174,145}
\definecolor{Reds5seq3}{RGB}{251,106,74}
\definecolor{Reds5seq4}{RGB}{222,45,38}
\definecolor{Reds5seq5}{RGB}{165,15,21}

\allowdisplaybreaks

\let\nc\newcommand

\nc{\lset}{\left\{\left.}
\nc{\rset}{\right\}}
\nc{\lsetr}{\left\{\,}
\nc{\rsetr}{\right.\right\}}
\nc{\barr}{\;\rule{0pt}{9.5pt}\left|\;}

\usepackage[most,breakable]{tcolorbox}
\renewenvironment{boxed}[1][white]%
  {\expandafter\ifstrequal\expandafter{#1}{filled}{\begin{tcolorbox}[colback=MidnightBlue!70!black!70!TealBlue!2!white,colframe=MidnightBlue!70!black!70!TealBlue!30!white,breakable,enhanced,left=5.75pt,right=5.75pt,grow sidewards by=10pt]}{\begin{tcolorbox}[colback=white,colframe=gray!15,breakable,enhanced,left=5.75pt,right=5.75pt,grow sidewards by=10pt]}}%
  {\end{tcolorbox}}

\nc{\R}{\mathcal{R}}
\nc{\W}{\mathcal{W}}
\nc{\T}{\mathcal{T}}
\nc{\B}{\mathcal{B}}
\nc{\C}{\mathcal{C}}
\nc{\U}{\mathcal{U}}
\nc{\E}{\mathcal{E}}
\nc{\EE}{\mathscr{E}}
\nc{\K}{\mathcal{K}}
\nc{\V}{\mathcal{V}}
\nc{\X}{\mathcal{X}}
\nc{\F}{\mathcal{F}}
\nc{\G}{\mathcal{G}}
\nc{\D}{\mathcal{D}}
\nc{\Y}{\mathcal{Y}}

\nc{\M}{\mathcal{M}}
\nc{\N}{\mathcal{N}}
\nc{\I}{\mathcal{I}}
\nc{\Q}{\mathbb{Q}}
\nc{\RR}{\mathbb{R}}
\nc{\CC}{\mathbb{C}}

\nc{\HH}{\mathbb{H}}
\nc{\MM}{\mathbb{M}}
\nc{\NN}{\mathbb{N}}
\nc{\DD}{\mathbb{D}}

\nc{\id}{\mathbbm{1}}
\nc{\idc}{\mathrm{id}}

\nc{\norm}[2]{\left\lVert#1\right\rVert_{\,#2}}
\nc{\proj}[1]{\ket{#1}\!\bra{#1}}
\nc{\lnorm}[2]{\left\lVert#1\right\rVert_{\ell_{#2}}}

\let\oldproofname\proofname
\renewcommand{\proofname}{\rm\bf{\oldproofname}}

\makeatletter
\renewenvironment{proof}[1][\proofname]{\par
\pushQED{\qed}%
\normalfont \topsep6\p@\@plus6\p@\relax
\trivlist
\item\relax
{\bfseries  %<- font shape
#1\@addpunct{.}}\hspace\labelsep\ignorespaces %<- includes punctuation code
}{%
\popQED\endtrivlist\@endpefalse
}
\makeatother

\nc{\rhos}{\rho'}

\nc{\gm}{\mathbin\#}

\newcommand{\sPnorm}{P,=}
\newcommand{\sPsub}{P,\leq}
\newcommand{\sTnorm}{T,=}
\newcommand{\sTsub}{T,\leq}
\newcommand{\smoothing}[2]{#2,\,#1}
\NewDocumentCommand{\s}{m m O{\ve}}{%
 \lowercase{\def\sdist{#1}}
 \lowercase{\def\snorm{#2}}
 \IfEqCase{\sdist}{%
  {t}{%
   \IfEqCase{\snorm}{%
    {norm}{ \smoothing{\sTnorm}{#3} }%
    {n}{ \smoothing{\sTnorm}{#3} }%
    {sub}{ \smoothing{\sTsub}{#3}}%
    {s}{ \smoothing{\sTsub}{#3}}%
   }%
  }%
  {p}{%
   \IfEqCase{\snorm}{%
    {norm}{ \smoothing{\sPnorm}{#3} }%
    {n}{ \smoothing{\sPnorm}{#3} }%
    {sub}{ \smoothing{\sPsub}{#3}}%
    {s}{ \smoothing{\sPsub}{#3}}%
   }%
  }%
 }%
}%
\NewDocumentCommand{\Dmax}{m m O{\ve}}{%
 D_{\max}^{\s{#1}{#2}[#3]}
}%
\makeatletter

\def\@sect@ltx#1#2#3#4#5#6[#7]#8{%
    \@ifnum{#2>\c@secnumdepth}{%
        \def\H@svsec{\phantomsection}%
        \let\@svsec\@empty
    }{%
        \H@refstepcounter{#1}%
        \def\H@svsec{%
            \phantomsection
        }%
        \protected@edef\@svsec{{#1}}%
        \@ifundefined{@#1cntformat}{%
            \prepdef\@svsec\@seccntformat
        }{%
            \expandafter\prepdef
            \expandafter\@svsec
            \csname @#1cntformat\endcsname
        }%
    }%
    \@tempskipa #5\relax
    \@ifdim{\@tempskipa>\z@}{%
        \begingroup
        \interlinepenalty \@M
        #6{%
            \@ifundefined{@hangfrom@#1}{\@hang@from}{\csname @hangfrom@#1\endcsname}%
            {\hskip#3\relax\H@svsec}{\@svsec}{#8}%
        }%
        \@@par
        \endgroup
        \@ifundefined{#1mark}{\@gobble}{\csname #1mark\endcsname}{#7}%
        \addcontentsline{toc}{#1}{%
            \@ifnum{#2>\c@secnumdepth}{%
                \protect\numberline{}%
            }{%
                \protect\numberline{\csname the#1\endcsname}%
            }%
            #7}% <<<<<<<<<<<<<<<<<<<<<<<< changed from #8
    }{%
        \def\@svsechd{%
            #6{%
                \@ifundefined{@runin@to@#1}{\@runin@to}{\csname @runin@to@#1\endcsname}%
                {\hskip#3\relax\H@svsec}{\@svsec}{#8}%
            }%
            \@ifundefined{#1mark}{\@gobble}{\csname #1mark\endcsname}{#7}%
            \addcontentsline{toc}{#1}{%
                \@ifnum{#2>\c@secnumdepth}{%
                    \protect\numberline{}%
                }{%
                    \protect\numberline{\csname the#1\endcsname}%
                }%
                #8}%
        }%
    }%
    \@xsect{#5}}%

\makeatother

\newcommand{\uX}{X^n}
\newcommand{\uY}{Y^n}

\newcommand{\ly}[1]{{\color{red} #1}}

\newtheorem{definition}{Definition}

\begin{document}
\count\footins = 1000

\title{Strong converse bounds on the classical identification capacity of the qubit depolarizing channel}

\author{Liuhang Ye}
\email{ly400@cam.ac.uk}
\affiliation{Department of Applied Mathematics and Theoretical Physics, Centre for Mathematical Sciences, University of Cambridge, Cambridge CB3 0WA, United Kingdom}

\author{Bjarne Bergh}
\email{bb536@cam.ac.uk}
\affiliation{Department of Applied Mathematics and Theoretical Physics, Centre for Mathematical Sciences, University of Cambridge, Cambridge CB3 0WA, United Kingdom}

\author{Nilanjana Datta}
\email{n.datta@damtp.cam.ac.uk}
\affiliation{Department of Applied Mathematics and Theoretical Physics, Centre for Mathematical Sciences, University of Cambridge, Cambridge CB3 0WA, United Kingdom}

\begin{abstract}

A strong converse bound for the classical identification capacity of a quantum channel is an upper bound on the asymptotic identification rate of classical messages sent through the channel, such that, above this rate, the probability of an identification error necessarily converges to one. Converse bounds for identification are notoriously difficult to obtain for fully quantum channels. The only previously known converse bound, due to Atif, Pradhan  and Winter [Int.~J.~Quantum Inf.~22(5):2440013, 2024], has the unsatisfactory feature of remaining strictly positive even for a completely noisy channel, for which identification is clearly impossible. We derive strong (and hence also weak) converse bounds, for the qubit depolarizing channel with noise parameter $p$, that vanish as $p\to 1$, thereby yielding the correct behavior in the completely noisy limit. Moreover, in the setting of simultaneous classical identification under the constraint of complete product measurements, our converse bound matches the corresponding achievability bound, and establishes that in this case the identification capacity equals the classical capacity of the channel.

\end{abstract}

\maketitle

\tableofcontents

%%%%%%%%%%%%%%%%%%%%%%%%%%%%%%%%%%%%%%%%%%%%%%%%%%%%%%%%%%%%%%%%%%%%%%%%%%%%%%%%
%%%%%%%%%%%%%%%%%%%%%%
\section{Introduction}
A fundamental task of communication over a noisy channel is the reliable transmission of classical messages, where the receiver is required to reconstruct the transmitted message with vanishing error probability, in the limit $n\to \infty$, where $n$ denotes the number of parallel uses of the channel. When the channel is quantum, i.e.,given by a completely positive trace-preserving (CPTP) map, the maximum number of bits that can be reliably transmitted per channel use, in the limit $n \to \infty$, is called the {classical capacity} of the quantum channel. One of the cornerstones of quantum Shannon theory, %due to Holevo-Schumacher-Westmoreland 
is the Holevo–Schumacher–Westmoreland Theorem~\cite{holevo2002capacity, schumacher1997sending}, by which the 
%is the discovery of an expression of the 
classical capacity is given by the regularization of the so-called Holevo capacity of the channel.
%as regularized entropic quantities.

In contrast, Ahlswede and Dueck \cite{ahlswede2002identification} introduced a different task of communication over noisy channels, known as \emph{identification}, which was later extended to quantum channels by L\"ober \cite{lober1999quantum}. %The %paradigm of message identification considers a different operational objective: i
In this task, instead of decoding the message itself, the receiver is only required to determine, for any candidate message of interest, whether or not it coincides with the transmitted one. This relaxed requirement leads to a striking difference in the scaling behavior. While the maximal number of messages that can be {{transmitted}} reliably grows exponentially with the blocklength (i.e., number of channel uses) $n$, the maximum number of messages that can be {{identified}} reliaby %identification codes allow a number of messages that 
grows {doubly exponentially} with $n$. This phenomenon reveals a fundamentally different notion of information capacity and has led to a rich theory with connections to the approximation of output distributions~\cite{verdhapproximation}, operator Chernoff bound~\cite{ahlswede2002strong}, soft-covering and channel resolvability, see e.g.~\cite{hayashi2025resolvability,cheng2023error,atif2024quantum} and references therein. {The task of different versions of identification over both classical and quantum channels\footnote{These include works on both point-to-point channels and broadcast channels.} have been studied extensively; see e.g.~\cite{colomer2025zero, winter2006identification, bracher2017identification,rosenberger2023identification} and references therein.

In this paper, we focus on strong converse bounds for simultaneous-, as well as unrestricted classical identification capacities of a quantum channel (both to be defined later). To our knowledge, the only known results on such bounds %for classical identification over fully quantum channels 
are those given in~\cite{atif2024quantum}, and involve an additive term given by the logarithm of the minimum of the dimensions of the input and output Hilbert spaces of the channel. Due to the latter, these bounds remain strictly positive even for very noisy channels, over which identification is clearly impossible. Hence, these bounds are not tight.

We focus on strong converse bounds for the classical identification capacity of the qubit depolarizing channel $\mathcal{N}_p$ with noise parameter $p$: for any input state~$\rho$,$$\mathcal{N}_p(ρ) = (1 - p)ρ + p \id/2.$$ In the limit $p\to1$, our bounds reduce to zero, which is in accordance with the fact that a completely noisy channel cannot have a positive identification capacity.

 Our method might provide some insight on obtaining strong converse bounds on the classical identification capacity of more general channels.

%All relevant notation and definitions are given in~\zcref{sec:math prem}.

%%%%%%%%%%%%%%%%%%%%%%%%%%%%%%%%%%%%%%%%%%%%%%%%%%%%%%%%%%%%%%%%%%%%%%%%%%%%%%%%
\bigskip

\noindent
{\bf{Layout of the paper:}} In Section~\ref{sec:background} we give a brief review of the task of classical identification over quantum channels and known results on it. In Section~\ref{sec:summary} we give a summary of our main results. All necessary notations and definitions are given in Section~\ref{sec:math prem}. Section~\ref{sec:sim depolarizing converse} deals with the simultaneous identification of classical messages over a qubit depolarizing channel under the constraint of complete product measurements. It contains Theorem~\ref{thm: sim depolarizing strong converse} and its proof. In Section~\ref{sec:unrestricted converse} we obtain a strong converse bound (Theorem~\ref{thm: depolarizing strong converse}
) for identification over the qubit depolarizing channel, without the constraint of simultaneity. A strong converse bound for classical identification over general quantum channels is derived in Section~\ref{sec:general converse} (Theorem~\ref{thm: general covering bound}). We conclude the paper with a summary and some open questions in Section~\ref{sec: open question}.

\section{The task of identification and known results}\label{sec:background}
\subsection{Classical message identification over quantum channels}
Formally, the two tasks of message transmission and message identification are defined as follows, in terms of transmission codes and identification codes, respectively.

\begin{definition}
An $(n,M,\lambda)$ classical transmission code for the quantum channel $\mathcal{N}:A\to B$, with worst-case error probability $\lambda$, is a set $\{(\rho_i,D_i): i=1,\dots,M\}$ of states $\rho_i\in \mathcal{D}(\mathcal{H}_A^{\otimes n})$ and POVM elements
$D_i\in \mathcal{B}(\mathcal{H}_B^{\otimes n})$, that is 
\[
\forall i, \quad D_i \ge 0,
\qquad
\sum_{i=1}^M D_i = \id_{B^n},
\]
such that
\begin{equation}
\forall i \qquad
\Tr\!\left[\mathcal{N}^{\otimes n}(\rho_i) D_i\right] \ge 1 - \lambda .
\end{equation}

We denote the largest size $M$ of such a transmission code by $M(n,\lambda)$.
\end{definition}

\begin{definition}[(L\"ober \cite{lober1999quantum})]
\label{def: quantum ID}
An $(n,N,\lambda_1,\lambda_2)$ classical identification code for the quantum channel $\mathcal{N}:A\to B$, with Type-I error probability
$\lambda_1$ and Type-II error probability $\lambda_2$, is a set
$\{(\rho_i,D_i): i=1,\dots,N\}$ of states $\rho_i\in \mathcal{D}(\mathcal{H}_A^{\otimes n})$ and operators
$D_i\in \mathcal{B}(\mathcal{H}_B^{\otimes n})$ with $0 \le D_i \le \id$, i.e., the pair $(D_i,\id-D_i)$
forms a binary POVM, such that
\begin{align}
\forall i \qquad & \Tr\bigl[\mathcal{N}^{\otimes n}(\rho_i) D_i\bigr] \ge 1-\lambda_1, \\
\forall i \neq j \qquad & \Tr\bigl[\mathcal{N}^{\otimes n}(\rho_i) D_j\bigr] \le \lambda_2 .
\end{align}

An identification code as above is called \emph{simultaneous} if all the
$D_i$ are coexistent: this means that there exists a common refining POVM $(E_t)_{t=1}^T$ and subsets $I_i \subset \{1,\dots,T\}$ such that
\begin{equation}
D_i = \sum_{t \in I_i} E_t .
\end{equation}
We denote the largest size $N$ of such a (simultaneous) identification code by $N_{(\mathrm{sim})}(n,\lambda_1,\lambda_2)$.
\end{definition}

In contrast to classical message transmission over a quantum channel ${\cal{N}}$, where the decoding measurement must form a single POVM in order to produce an unambiguous conclusion, identification decoding does not demand such a constraint. The reason is that in transmission, the receiver is required to answer an $N$-ary question—namely, which message was sent? —where $N$ is the number of all possible messages. In the case of classical message identification over ${\cal{N}}$, however, the receiver only needs to answer a binary question of the form: is the transmitted message equal to a specific message $i$?
As a result, the decoding side of an identification code consists of $N$ independent binary tests, described by pairs of POVMs $(D_i,\id-D_i)$, rather than a single $N$-outcome POVM. In particular, it is \emph{not} required that $\sum_i D_i \le \id$. The decoder's choice of $i$ is arbitrary and may vary from one use to another, and the sender has to encode his message without knowledge of which message the decoder wants to test for. 

In the most general setting, the decoding measurements $\{(D_i,\id-D_i)\}_i$ may be incompatible for $i\neq j$, known as an {\emph{unrestricted}} code. Operationally, the receiver can only have one single (but arbitrary) message $i$ in mind and perform the corresponding binary measurement $(D_i,\id-D_i)$, interpreting the outcome associated with $D_i$ as acceptance. If the identification code is simultaneous, however, all decoding operators $\{D_i\}_i$ are compatible: they come from a common underlying POVM $\{E_t\}_{t=1}^T$, so that $D_i=\sum_{t\in I_i}E_t$, for some subset $I_i \subset \{1,\dots,T\}$. The receiver, upon receiving the channel output state, will perform the measurement specified by the POVM $\{E_t\}_{t=1}^T$, obtaining an outcome $t$, and can therefore simultaneously answer the identification question for any message $i$ by checking whether $t \in I_i$. Crucially, the subsets $I_i \subset \{1,\dots,T\}$ need not be disjoint: it may occur that a single outcome $t$ belongs to both $I_i$ and $I_j$ for $i \neq j$, in which case the receiver would answer ``yes'' to both identification questions. 

\subsection{Transmission capacity versus identification capacity}
The capacity of a certain communication task is a characterization of how the maximum number of messages one wishes to communicate scales with the number of uses of the channel, such that the communication is reliable (i.e., the error vanishes) in the asymptotic limit. 

It is well known that the number of classical messages $M$ that can be reliably transmitted through a quantum channel $\mathcal{N}$ scales exponentially with the blocklength $n$, where the exponent is given by the classical capacity $C(\mathcal{N})$ of the channel, i.e.,$M \sim 2^{nC({\cal{N}})}$. Equivalently,
\begin{equation}
    C(\mathcal{N})\coloneqq
\inf_{\lambda > 0}\;
\liminf_{n \to \infty}
\frac{1}{n} \log M(n,\lambda)\, ,
\end{equation}
here and henceforth $\log$ denotes logarithm with base~2. The Holevo–Schumacher–Westmoreland theorem~\cite{holevo2002capacity,schumacher1997sending} states that the classical capacity $C(\mathcal{N})$ of a channel is equal to its regularized Holevo capacity:
\begin{equation}
C(\mathcal{N})
=
\inf_{\lambda > 0}\;
\limsup_{n \to \infty}
\frac{1}{n} \log M(n,\lambda)\,
=
\lim_{n \to \infty} \frac{1}{n} \chi\!\left( \mathcal{N}^{\otimes n} \right),
\end{equation}
\begin{equation}
\chi(\mathcal{N})
=
\max_{\{p_x,\rho_x\}}
\left[
S\!\left( \sum_x p_x \mathcal{N}(\rho_x) \right)
-
\sum_x p_x S\!\left( \mathcal{N}(\rho_x) \right)
\right],
\end{equation}
where $S(\rho)=-\Tr(\rho\log\rho)$ denotes the von Neumann entropy. The classical capacity $C(\mathcal{N})$ can equivalently be expressed as
\begin{align}
   C(\mathcal{N})& =  \lim_{n \to \infty} \frac{1}{n}  \sup_{\omega_{X^nB^n}} I(X^n; B^n)_\omega,
   \label{Hcap}
\end{align}
where 
$\omega_{X^nB^n} = \sum_{x^n} p(x^n) \ketbra{x^n} \otimes {\mathcal{N}^{\otimes n}}(\rho_{x^n}).$

For message identification, surprisingly, the maximum number of messages $N$ one can (simultaneously) identify scales \emph{doubly exponentially} with $n$, where the double exponent is given by the (simultaneous) classical identification capacity $C^{(\mathrm{sim})}_{\mathrm{ID}}(\mathcal{N})$ of the channel:
\begin{align}
C_{\mathrm{ID}}(\mathcal{N})
&\coloneqq
\inf_{\lambda_1, \lambda_2 > 0}\;
\liminf_{n \to \infty}
\frac{1}{n} \log \log N(n,\lambda_1,\lambda_2),\\
C^{\mathrm{sim}}_{\mathrm{ID}}(\mathcal{N})
&\coloneqq
\inf_{\lambda_1, \lambda_2 > 0}\;
\liminf_{n \to \infty}
\frac{1}{n} \log \log N_{\mathrm{sim}}(n,\lambda_1,\lambda_2).
\end{align}

This doubly exponential scaling was first noticed by Ahlswede and Dueck \cite{ahlswede2002identification} by showing that one can always construct an identification code by first using a transmission code with input alphabet size $|\mathcal{M}|\sim 2^{nC}$ and then using an identification code for the identity channel established by the transmission code, and this identification code will have size exponential in $|\mathcal{M}|$ and hence doubly exponential in $n$. For quantum channels this was first noted by L\"ober and leads to the following achievability bounds on the classical identification capacity over quantum channels.

\begin{boxed}
\begin{theorem}[(L\"ober \cite{lober1999quantum})]
\label{thm: achiev}
For a quantum channel $\mathcal{N}$,
\begin{equation}
    C_{\mathrm{ID}}(\mathcal{N}) \ge C^{\mathrm{sim}}_{\mathrm{ID}}(\mathcal{N}) \ge C(\mathcal{N}).
\end{equation}
\end{theorem}
\end{boxed}

\subsection{Known converse bounds for identification capacities}
Given a channel $\mathcal{N}$, the quantity $\hat{C}_{\mathrm{ID}}(\mathcal{N})$ is called a strong converse identification bound, if the following holds: let $\lambda_1, \lambda_2 > 0$ such that $\lambda_1 + \lambda_2 < 1$,\footnote{This is to exclude the trivial “always yes” code ($D_i=\id$ for all $i$) yielding $\lambda_1=0, \lambda_2=1$ and unbounded message size.} then for every $\delta > 0$ and sufficiently large $n$,
\begin{equation}
    N(n,\lambda_1,\lambda_2)\le 2^{2^{n\bigl(\hat{C}_{\mathrm{ID}}(\mathcal{N})+\delta\bigr)}},
\end{equation}
or equivalently,
\begin{equation}
C_{\mathrm{ID}}(\mathcal{N}) \le \limsup_{n \to \infty}
\frac{1}{n} \log \log N(n,\lambda_1,\lambda_2) \leq \hat{C}_{\mathrm{ID}}(\mathcal{N}).
\end{equation}
A strong converse bound for simultaneous identification is similarly defined.

For fully classical channels ($W_{cc}$), the identification capacity\footnote{Identification codes over fully classical channels are always simultaneous, as classical systems are not disturbed by measurements.} equals its Shannon capacity~\cite{ahlswede2002identification,verdhapproximation}: $C_\mathrm{ID}(W_{cc})=C_{\mathrm{Sh}}({W_{cc}}).$
For classical-quantum channels ($W_{cq}$), both unrestricted- and simultanoeus identification capacities are equal to its classical capacity~\cite{lober1999quantum, ahlswede2002strong}:
\begin{equation}
    C_\mathrm{ID}(W_{cq})=C^{\mathrm{sim}}_\mathrm{ID}(W_{cq})=C(W_{cq}).
\end{equation}
In both cases, the strong converse property holds, i.e.,allowing for non-vanishing errors $λ_1, λ_2$ does not increase the identification capacity.

In the fully quantum case, an explicit expression for the classical identification capacities is only known for the identity channel $(\mathrm{id}_A)$:
%however, no matching converse bound is known in general. The main difficulty lies in the arbitrary choice of encoding states, which includes not only the freedom in choosing input distributions, but also the corresponding bases. The only known matching bound is that for a noiseless channel $(\mathrm{id}_A)$:
\begin{equation}
    C^{\mathrm{sim}}_\mathrm{ID}(\mathrm{id}_A)=\log |A|, \quad C_\mathrm{ID}(\mathrm{id}_A)=2\log |A|,
\end{equation}
and each of these capacities satisfy the strong converse property.

The best known strong converse bounds for the simultaneous- and unrestricted classical identification capacities of a general quantum channel $\mathcal{N}$ to date are based on a quantum soft-covering lemma~\cite{atif2024quantum} and given by:
\begin{equation}
\label{Q-soft-covering bound}
    C_{\mathrm{ID}}^{\mathrm{sim}}(\mathcal{N}) \le \log \min \{|A|,|B|\},\quad C_{\mathrm{ID}}(\mathcal{N}) \le \log |A|+\hat{Q}(\mathcal{N}),
\end{equation}
where $|A|,|B|$ are dimensions of the channel input and output spaces, $\hat{Q}(\mathcal{N})$ is the strong converse quantum capacity.

We briefly explain how to establish converse bounds for identification. Suppose that $\bigl\{ (\rho_i, D_i) : i = 1,\dots, N \bigr\}$ is an $(n,N,\lambda_1,\lambda_2)$ identification code for $\mathcal{N}$. It follows from \zcref{def: quantum ID} that
\begin{equation}
\forall\, i \neq j,\qquad
\Tr\bigl[\left(\mathcal{N}^{\otimes n}(\rho_i) - \mathcal{N}^{\otimes n}(\rho_j)\right)D_i\bigr] \ge 1 - \lambda_1 - \lambda_2.
\end{equation}
By the variational characterization of trace distance, it follows that
\begin{equation}
\forall\, i \neq j,\qquad
\frac{1}{2}\left\| \mathcal{N}^{\otimes n}(\rho_i) - \mathcal{N}^{\otimes n}(\rho_j) \right\|_1 = \max_{ 0\le \Lambda \le \id} \Tr\bigl[(\mathcal{N}^{\otimes n}(\rho_i) - \mathcal{N}^{\otimes n}(\rho_j))\Lambda\bigr] \ge 1 - \lambda_1 - \lambda_2.
\label{sep req}
\end{equation}
This is a necessary condition for the existence of an identification code, reflecting the requirement that the channel outputs corresponding to distinct codewords are well separated. Consequently, a strong converse bound can be obtained by upper bounding the covering number of the output space of the channel. In particular, the number of messages $N(n,\lambda_1,\lambda_2)$ cannot exceed the minimum covering number; otherwise, two distinct outputs would lie close to the same element in the covering set, leading to a violation of~(\ref{sep req}). We make this argument precise in later sections.

\section{Summary of main results}
\label{sec:summary}
This paper focuses primarily on strong converse bounds for both unrestricted and simultaneous identification over the qubit depolarizing channel $\mathcal{N}_p$. In \zcref{sec:general converse}, we additionally establish a strong converse bound for general quantum channels.

\smallskip
\noindent
$\bullet$ We first study the simultaneous identification capacity with a restriction on the decoders. We prove that the simultaneous identification capacity of the depolarizing channel, under the constraint that the decoding POVM $\{E_t\}_t$ is a \emph{complete product measurement} (which we denote as  $\tilde{C}_{\mathrm{ID}}^{\mathrm{sim}}$), is equal to its classical capacity and a strong converse holds: $\forall\, \lambda_1, \lambda_2 >0$ such that $\lambda_1 + \lambda_2 < 1$,
\begin{equation}
    \tilde{C}_{\mathrm{ID}}^{\mathrm{sim}}(\mathcal{N}_p) = \lim_{n \to \infty}
\frac{1}{n} \log \log \tilde{N}_{\mathrm{sim}}(n,\lambda_1,\lambda_2) =C(\mathcal{N}_p)=1-h(p/2),
\label{matching bound}
\end{equation}
where $h$ is the binary entropy function: $h(p) \coloneqq - p \log p - (1-p) \log (1-p)$. This result is established by reducing the input states of the channel to classical probability distributions, and then using known soft-covering results~\cite{cheng2023error} to cover output probability distributions; see~\zcref{thm: sim depolarizing strong converse} and FIG 1.~(B).
\smallskip

\noindent
    $\bullet$ We then turn to the unrestricted identification capacity (i.e., without the simultaneity constraint), for which achievability bounds are known which imply that the identification capacity is in general larger than the classical capacity \cite{hayden2012weak, winter2013identification, winter2004quantum}. However, the only converse bounds that are known have a strong dependency on the dimension of the channel's input and output spaces, and are tight only for the identity channel. In particular, these bounds do not decrease to zero when the channel becomes increasingly noisy. Here - for the depolarizing channel - we provide a (strong) converse bound which decreases to zero as the depolarizing error probability $p\to1$. Concretely, we prove the following:
    \begin{equation}
    \label{result_5}
    C_{\mathrm{ID}}(\mathcal{N}_p) \leq 
    \begin{cases}
    2,
    & 0 \le p \le 1 - 2^{-2/3}, \\[0.6em]
    2 - D\!\left(\gamma(p)\,\middle\|\,\frac{3}{4}\right),
    & 1 - 2^{-2/3} \le p < 1,
    \end{cases}
    \quad \text{where} \quad \begin{tabular}{l}$ \displaystyle \gamma(p) \coloneqq \frac{-1}{2 \log(1-p)},$ \\ $\displaystyle  D(x\|y)
    \coloneqq x \log\frac{x}{y}
    + (1-x)\log\frac{1-x}{1-y}.$ \end{tabular}
    \end{equation}
    
    This result is established by directly covering the output geometry of the depolarizing channel\footnote{That is, the space of all output states of $\mathcal{N}^{\otimes n}_p$, for any fixed integer $n$.}, using known results on the covering number of ellipsoids~\cite{dumer2006covering}. This provides the first strong converse bound,  for classical identification over a fully quantum channel, that avoids the $\log$-dimension factor; see~\zcref{thm: depolarizing strong converse} and FIG 1.~(A).

\smallskip

\noindent
    $\bullet$ Lastly, for a {\em{general quantum channel}} $\mathcal{N}$, we show that one can replace the strong converse quantum capacity $\hat{Q}(\mathcal{N})$ appearing in the quantum soft-covering bound of~\cite{atif2024quantum} (given by (\ref{Q-soft-covering bound})), with the classical capacity $C(\mathcal{N})$:
\begin{equation}
    C_{\mathrm{ID}}(\mathcal{N}) \le \log |A|+C(\mathcal{N}),
    \label{general converse bound}
\end{equation}
where $|A|$ is the dimension of the channel's input space. We obtain this result by employing the classical-quantum soft-covering lemma~\cite{cheng2023error}, to approximate any output state with an input state with eigenvalues given by an $M$-type where $M\sim2^{nC(\mathcal{N})}$; see~\zcref{thm: general covering bound} for details. This bound is, in general, more useful than %easier to compute than 
(\ref{Q-soft-covering bound}) for the following reasons: (i) unlike the classical capacity, no explicit expression for the strong converse quantum capacity is known for a general quantum channel; (ii) for channels with an additive Holevo capacity (e.g.~the depolarizing channel) the classical capacity can be evaluated explicitly. For large values of the parameter $p$ (i.e.,$p >0.82$), our bound in (\ref{result_5}) is an improvement over the corresponding bound for the qubit depolarizing channel obtained from (\ref{general converse bound}); see FIG 1. (A).
\smallskip

%The above results are illustrated in Figure~\ref{fig:cid-depolarizing}.

\begin{figure}[htbp]
    \centering
    \includegraphics[width=1\linewidth]{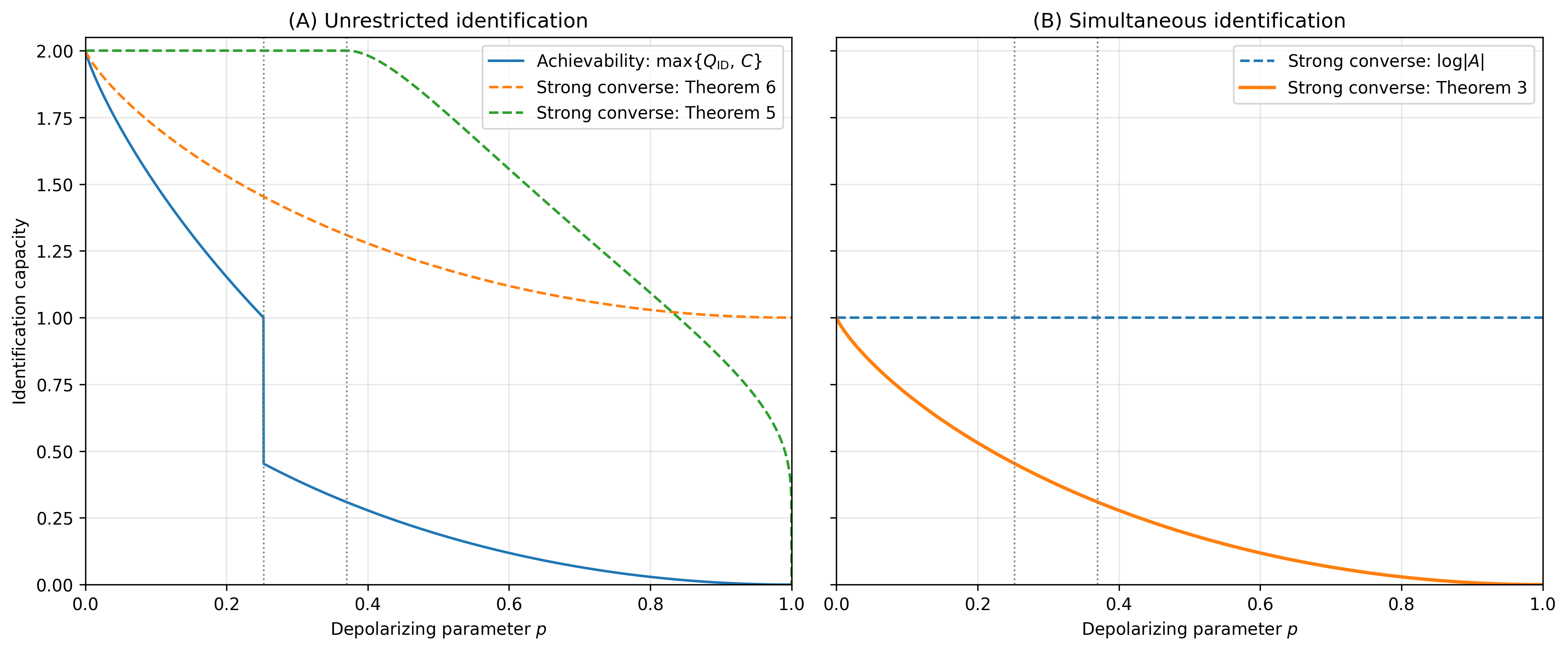}
    \caption{{(A): Unrestricted identification codes} The orange and green curves show the strong converse bounds~(\ref{general converse bound}) for $\mathcal{N}_p$ and~(\ref{result_5}), respectively. {(B): Simultaneous identification codes with complete product measurements.} The orange curve gives the achievability bound and the matching strong converse bound~(\ref{matching bound}), while the blue curve shows the previous $\log$-dimension upper bound (\ref{Q-soft-covering bound}) of \cite{atif2024quantum}.}
    \label{fig:cid-depolarizing}
\end{figure}

%\textcolor{red}{Move this to somewhere else:} Whether the unrestricted simultaneous identification capacity is also equal to the classical capacity is an open question:
%\begin{equation}
%    C_{\mathrm{ID}}^{\mathrm{sim}}(\mathcal{N}_p) \stackrel{?}{=} C(\mathcal{N}_p).
%\end{equation}

\section{Mathematical preliminaries}
\label{sec:math prem}
\subsection{Notations and Definitions}
Let $\mathcal{H}$ denote a complex finite-dimensional Hilbert space, and $\mathcal{B}(\mathcal{H})$ the set of linear operators acting on $\mathcal{H}$. Let $\mathcal{D}(\mathcal{H})$ denote the set of density matrices on $\mathcal{H}$, i.e., the set of positive semidefinite operators of unit trace. A quantum channel (usually denoted by $\mathcal{N}$) is a completely positive trace-preserving map acting on $\mathcal{B}(\mathcal{H})$. We label different quantum systems by capital Roman letters ($A,B,C$, etc.) and often use these letters interchangeably with the corresponding Hilbert space or set of density matrices (i.e., we write $\rho \in \mathcal{D}(A)$ instead of $\rho \in \mathcal{D}(\mathcal{H}_A)$ and $\mathcal{N}: A \to B$ instead of $\mathcal{N}: \mathcal{B}(\mathcal{H}_A) \to \mathcal{B}(\mathcal{H}_B)$). We also concatenate these letters to denote tensor products of systems, i.e., we write $\rho \in \mathcal{D}(RA)$ for $\rho \in \mathcal{D}(\mathcal{H}_R \otimes \mathcal{H}_A)$. The identity operator in $\mathcal{B}(\mathcal{H})$ is denoted by $\id$, whereas the identity map acting on $\mathcal{B}(\mathcal{H})$ is denoted by $\mathcal{I}$. A positive operator-valued measure (POVM) on a Hilbert space $\mathcal{H}$ is a collection of positive semidefinite operators $\{E_t\}_{t=1}^T$ satisfying the normalization condition $\sum_{t=1}^T E_t = \id$.

The total variation (TV) distance between two probability distributions $p,q$ on a finite alphabet $\mathcal{X}$ is defined as $\|p - q\|_{\mathrm{TV}}
:= \frac{1}{2} \sum_{x \in \mathcal{X}} |p(x) - q(x)|$. The trace distance between any two states $\rho,\sigma \in \mathcal{D}(\mathcal{H})$ is given by $\frac{1}{2}\|\rho - \sigma\|_1$, where $\|A\|_1= \operatorname{Tr}(|A|)
= \operatorname{Tr}\!\left(\sqrt{A^\dagger A}\right)$.
The Hamming distance between two strings $x^n,y^n \in \mathcal{X}^n$ is defined as
\[
d_H(x^n,y^n)
:= \left|\left\{ i \in \{1,\dots,n\} : x_i \neq y_i \right\}\right|.
\]

\begin{definition}
\label{def: BSC}
Let $\mathcal{X} = \mathcal{Y} = \{0,1\}$ and let $p \in [0,1]$. The binary symmetric channel (BSC) with crossover probability $p$, denoted by $\mathrm{BSC}_p$, is defined by the transition probabilities
\begin{equation}
\mathrm{BSC}_p(y|x) =
\begin{cases}
1-p, & y = x, \\
p, & y \ne x,
\end{cases}
\qquad x,y \in \{0,1\}.
\end{equation}
\end{definition}

\begin{definition}
Let $p \in [0,1]$. The qubit depolarizing channel with error probability $p$, denoted by $\mathcal{N}_p: \mathcal{D}(\mathbb{C}^2) \to \mathcal{D}(\mathbb{C}^2)$, is defined by
\begin{equation}
    \forall \rho \in \mathcal{D}(\mathbb{C}^2), \quad \mathcal{N}_p(\rho) = \left( 1-p \right) \rho+p\Tr (\rho)\frac{\id_2}{2}.
\end{equation}
\end{definition}

\begin{definition}[($M$-type distribution)]
Let $\mathcal{X}$ be a finite alphabet and let $M \in \mathbb{Z}_+$. A probability mass function $p \in \mathcal{P}(\mathcal{X})$ is called an $M$-type if
\begin{equation}
p(x) \in \left\{0, \frac{1}{M}, \frac{2}{M}, \dots, \frac{M}{M}\right\}, \qquad \forall x \in \mathcal{X}.
\end{equation}
The set of all $M$-types on $\mathcal{X}$ is denoted by $\mathcal{P}_M(\mathcal{X})$. The cardinality of all $M$-types is upper bounded by the number of strings of length $M$: $|\mathcal{P}_M(\mathcal{X})|\le|\mathcal{X}|^M$.
\end{definition}

\subsection{Quantum divergences and mutual informations}
For two quantum states $\rho, \sigma \in {\cal{D}}({\cal{H}})$, where $\cal{H}$ is a finite-dimensional Hilbert space of dimension $d$, the {\em{sandwiched R\'enyi divergence}} is defined for any $\alpha >0$ and $\alpha \ne 0$ (~\cite{wilde2014strong} and~\cite{muller2013quantum})
\begin{align}
    \widetilde{D}_\alpha (\rho||\sigma) &:= \frac{1}{\alpha -1} \log \Bigl( \Tr \bigl(\sigma^{\frac{1-\alpha}{2\alpha}} \rho \sigma^{\frac{1-\alpha}{2\alpha}}\bigr)^\alpha\Bigr) 
\quad {\hbox{if}}\,\ {\rm{supp}}\,\rho \subseteq {\rm{supp}}\,\sigma,\end{align}
and $\widetilde{D}_\alpha (\rho||\sigma)= \infty$ if ${\rm{supp}}\,\rho \not\subseteq {\rm{supp}}\,\sigma$. 
This is a non-commutative generalization of the Petz R\'enyi divergence~\cite{petz1986quasi}
\begin{align}
    {D}_\alpha (\rho||\sigma) &:= \frac{1}{\alpha -1} \log \Bigl( \Tr \bigl(\rho^\alpha \sigma^{1-\alpha}\bigr)\Bigr) \quad {\hbox{if}}\,\ {\rm{supp}}\,\rho \subseteq {\rm{supp}}\,\sigma,
    \end{align}
    In the limit $\alpha \to 1$, both these quantities reduce to the Umegaki relative entropy~\cite{umegaki1962conditional}
\begin{align}
D(\rho||\sigma)&:= {\rm{Tr}} \rho (\log \rho  - \log \sigma) \quad {\hbox{if}}\,\ {\rm{supp}}\,\rho \subseteq {\rm{supp}}\,\sigma.
\end{align}
If $\rho$ and $\sigma$ commute, then the sandwiched R\'enyi divergence reduces to the Petz R\'enyi divergence~\cite{petz1986quasi}.
In this case $\rho$ and $\sigma$ are classical, i.e.,they are diagonal in the same basis, and if their diagonal entries are given by $\{p_i\}_{i=1}^d$ and $\{q_i\}_{i=1}^d$, then 
the above quantity is equal to the classical R\'enyi divergence~\cite{renyi1961measures} of two probability distributions $P=\{p_i\}_{i=1}^d$ and $Q=\{q_i\}_{i=1}^d$, 
\begin{align}
    {D}_\alpha (\rho||\sigma) \equiv D_\alpha (P||Q) = \frac{1}{\alpha -1} \log \Bigl( \sum_{i=1}^d 
    p_i^\alpha q_i^{1-\alpha}\Bigr) 
\end{align}
For a bipartite state $\rho_{AB} \in  {\cal{D}}({\cal{H}}_A \otimes {\cal{H}}_B)$ the {\em{order-$\alpha$ sandwiched R\'enyi mutual information}} $\widetilde{I}_\alpha(A;B)_\rho$, and the {\em{order-$\alpha$ Petz R\'enyi mutual information}} ${I}_\alpha(A;B)_\rho$ are, respectively, defined as 
\begin{align}
   \widetilde{I}_\alpha(A;B)_\rho &:= \inf_{\sigma_B \in {\cal{D}}({\cal{H}}_B)}  \widetilde{D}_\alpha (\rho_{AB}||\rho_A \otimes \sigma_B),\label{alpha sandwich mutual info}\\
    {I}_\alpha(A;B)_\rho &:= \inf_{\sigma_B \in {\cal{D}}({\cal{H}}_B)}  D_\alpha (\rho_{AB}||\rho_A \otimes \sigma_B).
\end{align}
A bipartite quantum state $\rho_{XB} \in  {\cal{D}}({\cal{H}}_X\otimes {\cal{H}}_B)$ 
is said to be a classical-quantum (CQ) state if it can be expressed in the form
\begin{align}
\rho_{XB}
&=
\sum_{x \in \mathcal{X}} p_X(x)\,
\ket{x}\!\bra{x}
\otimes
\rho_B^x,
\end{align}
where $\{\ket{x}\}_{x \in {\cal{X}}}$ denotes an orthonormal basis of ${\cal{H}}_X$, ${\cal{X}}$ is a finite alphabet, $\{p_X(x)\}_{x \in {\cal{X}}}$ is a probability distribution on ${\cal{X}}$, and for each $x \in {\cal{X}}$, $\rho_B^x \in {\cal{D}}({\cal{H}}_B)$. 

The state $\rho_{XB}$ reduces to a purely classical state if $[\rho_B^x , \rho_B^y ]=0$ for all $x,y \in  {\cal{X}}$. In this case, $\rho_{XB}$ is given by a classical probability distribution $P_{XY}$, with $X, Y$ being random variables with alphabet ${\cal{X}}$ and joint distribution $P_{XY}$:
\begin{align}
    \rho_{XB} = \sum_{x, y\in {\cal{X}}} P_{XY}(x,y) \ket{x}\!\bra{x}
\otimes \ket{y}\!\bra{y}\label{eq:c-c-state}
\end{align}
For such a state, the {{order-$\alpha$ sandwiched R\'enyi mutual information}} and the  {{order-$\alpha$ Petz R\'enyi mutual information}} reduce to the classical {\em{Sibson $\alpha$-mutual information}}~\cite{verdu2015alpha}, which for two random variables $X$ and $Y$, with joint distribution $P_{XY}$, is defined as follows, for $\alpha \in (0,1) \setminus \{1\}$:
\begin{align}
   {I}_\alpha^S(X;Y) &:= \inf_{{Q_Y}}D_\alpha(P_{XY}||P_XQ_Y),\label{eq:sibson}
   \end{align}
the infimum being over all probability distributions $Q_Y$ of the random variable $Y$, and $P_{XY} << P_XQ_Y$, i.e.,$P_X(x)Q_Y(y) = 0$ $\implies$ $P_{XY}(x,y) =0$, with $P_X$ denoting the marginal distribution of $X$. In the limit $\alpha \to 1$, it reduces to the classical mutual information $I(X;Y)= D(P_{XY}||P_XP_Y)$, where $P_X$ and $P_Y$ are the marginal distributions of $X$ and $Y$.
\bigskip

\subsection{Classical-quantum soft-covering}
We use the following theorem proved by Cheng and Gao~\cite{cheng2023error}) in our analysis:

\begin{boxed}
\begin{theorem}[(Classical-quantum soft covering~\cite{cheng2023error})]
\label{thm:soft_covering}
Consider the classical–quantum state
\begin{equation}
    \rho_{XB}
=
\sum_{x \in \mathcal{X}} p_X(x)\,
\ket{x}\!\bra{x}
\otimes
\rho_B^x.
\end{equation}
Let $\mathcal{C}=\{x(1),\dots,x(M)\}$ be a random codebook of size $M$, whose entries are drawn independently according to $p_X$, and define the codebook-induced average state
\begin{equation}
\rho_B^{\mathcal{C}}\coloneqq\frac{1}{M}
\sum_{m=1}^M \rho_B^{x(m)}.    
\end{equation}
Then for every $\alpha \in (1,2)$,
\begin{equation}
    \frac{1}{2}\,
\mathbb{E}_{\mathcal{C}}
\bigl\|
\rho_B^{\mathcal{C}} - \rho_B
\bigr\|_1
\le
2^{\frac{2}{\alpha}-2+
\frac{\alpha-1}{\alpha}
\bigl(
\tilde{I}_\alpha(X:B)_\rho - \log M
\bigr)},
\label{eq:RHSbound}
\end{equation}
with $\tilde{I}_\alpha(X:B)_\rho$ the order-$\alpha$ sandwiched Rényi mutual information, defined in~(\ref{alpha sandwich mutual info}).
\end{theorem}
\end{boxed}
\smallskip

In the proof of \zcref{thm: sim depolarizing strong converse}, we employ the above result but for the case in which $\rho_{XB}$ is a purely classical state, i.e.,it is of the form given by (\ref{eq:c-c-state}).
In this case, the bound in (\ref{eq:RHSbound}) reduces to the following:  for every $\alpha \in (1,2)$,
\begin{align}
    \frac{1}{2}\,
\mathbb{E}_{\mathcal{C}}
\bigl\|
\rho_B^{\mathcal{C}} - \rho_B
\bigr\|_1
&\le
2^{\frac{2}{\alpha}-2+
\frac{\alpha-1}{\alpha}
\bigl(
I_\alpha^S(X:Y) - \log M
\bigr)},
\label{eq:clbound}
\end{align}
where $I_\alpha^S(X:Y)$ is the Sibson $\alpha$-mutual information defined through (\ref{eq:sibson}).

%%%%%%%%%%%%%%%%%%%%%%%%%%%%%%%%%%%%%%%%%%%%%%%%%%%%%%%%%%%%%%%%%%%%%%%%%%%%%%%%
%%%%%%%%%%%%%%%%%%%%%%%%%%%%%%%%%%%%%%%%%%%%%%%%%%%%%%%%%%%%%%%%%%%%%%%%%%%%%%%%
\section{Simultaneous identification capacity under complete product measurements}\label{sec:sim depolarizing converse}
Recall that in simultaneous identification, the decoders $D_i$'s are the coarse-graining of a common POVM $\{E_t\}_{t=1}^T$: there exist subsets $I_i \subset \{1,\dots,T\}$ such that $D_i = \sum_{t \in I_i} E_t$. Therefore, applying the decoders to the channel output is the same as applying a measurement and then decoding according to classical decision regions $\{I_i\}_i$. The difference between a simultaneous identification code for a quantum channel and a purely classical identification code for a classical channel is that the encoders for a simultaneous code remain quantum, for the case of a quantum channel. In this Section, we will see that, for the case of the $n$-fold qubit depolarizing channel, imposing certain structure on the decoding measurements, we can also restrict the input to  $\mathcal{N}^{\otimes n}_p$ to be classical, i.e.,diagonal in a basis determined by the measurement.

\subsection{Reduction to the binary symmetric channel}
Let $\mathcal{H}=(\mathbb{C}^2)^{\otimes n}$ be the $n$-qubit Hilbert space on which the depolarizing channel $\mathcal{N}^{\otimes n}_p$ acts. A complete product measurement on $n$ qubits is a tensor product of local rank-1 projective measurements on individual qubits. For each $i\in[n]$, fix an arbitrary orthonormal basis
\begin{equation}
    \mathcal{B}_i=\{\ket{\psi_i},\ket{\psi_i^\perp}\}.
\end{equation}
This induces the product basis $\mathcal{B}=\{\ket{\Psi_{x^n}}:x^n\in\{0,1\}^n\}$ defined by
\begin{equation}
\ket{\Psi_{x^n}}\coloneqq\bigotimes_{i=1}^n \ket{\psi_i^{(x_i)}},
\qquad
\ket{\psi_i^{(0)}}=\ket{\psi_i},\quad
\ket{\psi_i^{(1)}}=\ket{\psi_i^\perp}.
\end{equation}

In general, a complete product measurement on the $n$-fold Hilbert space $\mathcal{H}^{\otimes n}$ with $\dim \mathcal{H}=d$ corresponds to a choice of a product basis $\mathcal{B}$ of $\mathcal{H}^{\otimes n}$, and is of the form:
\begin{equation}
M=\{E_{x^n}\}_{{x^n}\in[d]^n},\qquad
E_{x^n}:=\ket{\Psi_{x^n}}\!\bra{\Psi_{x^n}},
\end{equation}
where $\ket{\Psi_{x^n}}\in\mathcal{H}^{\otimes n}$ is a product state. Imposing that the decoding measurements of a simultaneous identification code are complete product measurements defines a restricted class of simultaneous identification codes, and the corresponding capacity.

\begin{definition}[(Identification with complete product measurement)]
\label{def: complete product ID}
An $(n,N,\lambda_1,\lambda_2)$ simultaneous classical identification code for the quantum channel $\mathcal{N}:A\to B$ is called an identification code with complete product measurement, if there exists a complete product measurement that serves as the common refining POVM:
\begin{equation}
M=\{E_{x^n}\}_{{x^n}\in[d_B]^n},\qquad
E_{x^n}:=\ket{\Psi_{x^n}}\!\bra{\Psi_{x^n}},
\end{equation}
where $d_B=\dim \mathcal{H}_B$, and $\{\ket{\Psi_{x^n}}\}_{x^n}$ forms a product basis of $\mathcal{H}_B^{\otimes n}$. Explicitly, there exist subsets $I_i \subset \{1,\dots,d_B\}^n$ such that the decoders of the identification code are given by
\begin{equation}
D_i = \sum_{{x^n} \in I_i} E_{x^n}.
\end{equation}

We denote the largest size $N$ of such a simultaneous identification code with complete product measurement by $\tilde{N}_{\mathrm{sim}}(n,\lambda_1,\lambda_2)$. The simultaneous identification capacity under the constraint of complete product measurements, denoted as $\tilde{C}_{\mathrm{ID}}^{\mathrm{sim}}(\mathcal{N})$, is then defined as:
\begin{equation}
    \tilde{C}^{\mathrm{sim}}_{\mathrm{ID}}(\mathcal{N})
\coloneqq
\inf_{\lambda_1, \lambda_2 > 0}\;
\liminf_{n \to \infty}
\frac{1}{n} \log \log \tilde{N}_{\mathrm{sim}}(n,\lambda_1,\lambda_2).
\end{equation}
\end{definition}

Our goal is to study the simultaneous identification capacity of the qubit depolarizing channel, under the constraint of complete product measurements, and obtain a single-letter formula for $\tilde{C}^{\mathrm{sim}}_{\mathrm{ID}}(\mathcal{N}_p)$, for which the strong converse property holds. The key observation, summarized in the following lemma, is that the probabilistic outcome of a complete product measurement on the output state of the depolarizing channel $\mathcal{N}^{\otimes n}_p(\rho)$, can be equivalently obtained by applying an $n$-fold binary symmetric channel (BSC) to the diagonal elements of the input state $\rho$. Let $\Delta_{\mathcal{B}}$ denote the completely dephasing channel in the basis $\mathcal{B}=\{\ket{\Psi_{x^n}}:x^n\in\{0,1\}^n\}$:
\begin{equation}
\Delta_{\mathcal{B}}(\rho)
\coloneqq
\sum_{x^n\in\{0,1\}^n}
\bra{\Psi_{x^n}}\rho\ket{\Psi_{x^n}}\;
\ket{\Psi_{x^n}}\!\bra{\Psi_{x^n}}.
\end{equation}
\begin{lemma}[(Reduction to an $n$-fold BSC)]
\label{lem:reduction-bsc}
Fix a complete product measurement corresponding to some product basis $\mathcal{B}=\{\ket{\Psi_{x^n}}\coloneqq\bigotimes_{i=1}^n \ket{\psi_i^{(x_i)}}:x^n\in\{0,1\}^n\}$, so that $\{E_{x^n}:=\ket{\Psi_{x^n}}\!\bra{\Psi_{x^n}}\}_{x^n}$. For any $n$-qubit input state $\rho$, this measurement on the corresponding output state $\mathcal{N}_p^{\otimes n}(\rho)$ has outcome probabilities: %define the measurement outcome after the depolarizing channel:
\begin{equation}
q_\rho(x^n)\coloneqq\Tr\!\big(E_{x^n}\,\mathcal{N}_p^{\otimes n}(\rho)\big),
\qquad x\in\{0,1\}^n.
\end{equation}
Then:

\begin{enumerate}
\item[(i)] The output distribution $q_\rho(x^n)$ depends only on the diagonal elements of $\rho$ in the basis $\mathcal{B}$, i.e.
\begin{equation}
q_\rho(x^n)
=
\Tr\left(E_{x^n}\,\mathcal{N}_p^{\otimes n}(\Delta_{\mathcal{B}}(\rho))\right).
\end{equation}

\item[(ii)] Define the input distribution $r_\rho(x^n)\coloneqq\bra{\Psi_{x^n}}\rho\ket{\Psi_{x^n}}, x^n\in \{0,1\}^n$, as the diagonal elements of $\rho$ in the basis $\mathcal{B}$. Then the output distribution $q_\rho(x^n)$ is obtained by passing $r(x^n)$ through an $n$-fold binary symmetric channel with crossover probability $p/2$, i.e.
\begin{equation}
q_\rho(x^n)
=
\left(\mathrm{BSC}_{p/2}^{\otimes n}(r_\rho)\right)(x^n)\coloneqq\sum_{y^n\in\{0,1\}^n}r_\rho(y^n)\mathrm{BSC}_{p/2}^{\otimes n}(x^n|y^n),
\end{equation}
where a single binary symmetric channel with crossover probability $p/2$ acts as (see~\zcref{def: BSC}):
\begin{equation}
\mathrm{BSC}_{p/2}(x|y)=
\begin{cases}
1-p/2, & x=y,\\
p/2, & x\neq y.
\end{cases}
\end{equation}
Equivalently, for all $x^n\in\{0,1\}^n$,
\begin{equation}
q_\rho(x^n)
=
\sum_{y^n\in\{0,1\}^n}
r_\rho(y^n)\,
(p/2)^{d_H(x^n,y^n)}(1-p/2)^{n-d_H(x^n,y^n)},
\end{equation}
where $d_H(x^n,y^n)$ is the Hamming distance between $x^n$ and $y^n$.
\end{enumerate}
\end{lemma}

\begin{proof}
(i) Write the depolarizing channel as the probabilistic mixture of the identity channel $\mathcal{I}$ and the completely depolarizing channel $\mathcal{D}$:
\begin{equation}
\mathcal{N}_p=(1-p)\mathcal{I}+p\mathcal{D},
\qquad
\mathcal{D}(\rho)=\frac{\id_2}{2}\Tr(\rho).
\end{equation}
By linearity,
\begin{equation}
\mathcal{N}_p^{\otimes n}
=
\sum_{S\subseteq[n]}
(1-p)^{n-|S|}p^{|S|}
\left(
\bigotimes_{i\notin S}\mathcal{I}_i
\otimes
\bigotimes_{i\in S}\mathcal{D}_i
\right).
\end{equation}

Fix $x^n\in\{0,1\}^n$, the probability of measuring $x^n$ is given by
\begin{equation}
\label{eq:q-sum-unified}
q_\rho(x^n)
=\bra{\Psi_{x^n}}\mathcal{N}_p^{\otimes n}(\rho)\ket{\Psi_{x^n}} =
\sum_{S\subseteq[n]}
(1-p)^{n-|S|}p^{|S|}
\bra{\Psi_{x^n}}
\left(
\mathcal{I}_{[n]\setminus S}\otimes\mathcal{D}_S
\right)(\rho)
\ket{\Psi_{x^n}}.
\end{equation}

Let $T \coloneqq [n]\setminus S$. Then the Hilbert space factorizes as $\mathcal{H}^{\otimes n} = \mathcal{H}_T \otimes \mathcal{H}_S$, and the product basis vectors factorize accordingly: $\ket{\Psi_{x^n}}=\ket{\Psi_{x_T}} \otimes \ket{\Psi_{x_S}}$, where $x_T \in \{0,1\}^{|T|}$ and $x_S \in \{0,1\}^{|S|}$ denote the restrictions of the bit string $x^n \in \{0,1\}^n$ to the coordinate sets $T$ and $S$, respectively. Note that we have the identity
\begin{equation}
\left(\mathcal{I}_T\otimes\mathcal{D}_S\right)(X)
=
\Tr_S(X) \otimes \frac{\id_S}{2^{|S|}}.
\end{equation}
Therefore each term
\begin{equation}
\bra{\Psi_{x^n}}
\left(\mathcal{I}_T\otimes\mathcal{D}_S\right)(\rho)
\ket{\Psi_{x^n}}
=
\frac{1}{2^{|S|}}
\bra{\Psi_{x_T}}\Tr_S(\rho)\ket{\Psi_{x_T}},
\end{equation}
in the sum depends only on the diagonal entries of $\rho$ in the basis $\mathcal{B}$.

\bigskip
(ii) By definition,
\begin{equation}
    \Delta_{\mathcal{B}}(\rho)=
\sum_{y^n\in \{0,1\}^n} 
\bra{\Psi_{y^n}}\rho\ket{\Psi_{y^n}}\;
\ket{\Psi_{y^n}}\!\bra{\Psi_{y^n}}=
\sum_{y\in \{0,1\}^n} 
r_\rho(y^n)
\ket{\Psi_{y^n}}\!\bra{\Psi_{y^n}}.
\end{equation}
By (i) and linearity,
\begin{align}
q_\rho(x^n)
&=
\sum_{y^n\in \{0,1\}^n} r_\rho(y^n)
\Tr\!\big(
E_{x^n}\mathcal{N}_p^{\otimes n}(\ket{\Psi_{y^n}}\!\bra{\Psi_{y^n}})
\big).
\label{eq:1}
\end{align}
Since both the channel and measurement factorize,
\begin{equation}
\Tr\!\big(
E_{x^n}\mathcal{N}_p^{\otimes n}(\ket{\Psi_{y^n}}\!\bra{\Psi_{y^n}})
\big)
=
\prod_{i=1}^n
\Tr\!\left(
\ket{\psi_i^{(x_i)}}\!\bra{\psi_i^{(x_i)}}
\mathcal{N}_p(\ket{\psi_i^{(y_i)}}\!\bra{\psi_i^{(y_i)}})
\right).    
\end{equation}
A direct computation shows
\begin{align}
\Tr\!\left(
\ket{\psi_i^{(b)}}\!\bra{\psi_i^{(b)}}
\mathcal{N}_p\big(\ket{\psi_i^{(a)}}\!\bra{\psi_i^{(a)}}\big)
\right)
&=
\Tr\!\left(
\ket{\psi_i^{(b)}}\!\bra{\psi_i^{(b)}}
\left((1-p)\ket{\psi_i^{(a)}}\!\bra{\psi_i^{(a)}}+\frac{p}{2}\id_2\right)
\right)
\nonumber
\\
&=
(1-p)\delta_{ab}
+
\frac{p}{2}.
\end{align}
Hence,
\begin{align}
\Tr\!\big(
E_{x^n}\mathcal{N}_p^{\otimes n}(\ket{\Psi_{y^n}}\!\bra{\Psi_{y^n}})
\big)
&=
(p/2)^{d_H(x^n,y^n)}(1-p/2)^{n-d_H(x^n,y^n)}.
\label{eq:2}
\end{align}
Substituting (\ref{eq:2}) into the RHS of (\ref{eq:1})  completes the proof.
\end{proof}

\subsection{Classical soft-covering and identification converse}
Having reduced the action of the $n$-fold qubit depolarizing channel to an $n$-fold classical BSC, we can use results on classical (or more generally, classical-quantum) soft-covering, \zcref{thm:soft_covering}, to cover target output states using empirical distributions (i.e., types) as input. The number of distinct types will then provide an upper bound on the number of messages in identification codes. See the theorem below.

\begin{boxed}
\begin{theorem}[(Simultaneous identification capacity of $\mathcal{N}_p$ under  complete product measurements)]\label{thm: sim depolarizing strong converse}
The simultaneous identification capacity of the qubit depolarizing channel with error probability $p\in[0,1]$, under the constraint of complete product measurements, is given by
\begin{equation}
    \tilde{C}_{\mathrm{ID}}^{\mathrm{sim}}(\mathcal{N}_p)=C(\mathcal{N}_p)=1-h(p/2),
\end{equation}
and strong converse holds: for any $\lambda_1,\lambda_2 >0$, such that $\lambda_1+\lambda_2 <1$,
\begin{equation}
    \tilde{C}_{\mathrm{ID}}^{\mathrm{sim}}(\mathcal{N}_p) \le \limsup_{n \to \infty}
\frac{1}{n} \log \log \tilde{N}_{\mathrm{sim}}(n,\lambda_1,\lambda_2) \leq 1-h(p/2),
\end{equation}
where $h$ is the binary entropy function: $h(p) \coloneqq - p \log p - (1-p) \log (1-p).$
\end{theorem}
\end{boxed}

\begin{proof}
Fix $n\in\mathbb{N}$ and consider an arbitrary $(n,N,\lambda_1,\lambda_2)$ simultaneous identification code with complete product measurement $\{E_{x^n}=\ket{\Psi_{x^n}}\!\bra{\Psi_{x^n}}\}_{x^n\in\{0,1\}^n}$ associated with some product basis $\mathcal{B}=\{\ket{\Psi_{x^n}}\}_{x^n\in\{0,1\}^n}$. We call $\mathcal{M}$ the quantum-classical (QC) channel that performs the measurement in the basis $\mathcal{B}$. It then acts on the depolarizing output as
\begin{equation}
\mathcal{M}\circ\mathcal{N}^{\otimes n}_p(\rho)=\sum_{x^n} \Tr(E_{x^n}\mathcal{N}^{\otimes n}_p(\rho))\ket{x^n}\!\bra{x^n}\equiv\sum_{x^n} q_\rho(x^n)\ket{x^n}\!\bra{x^n}.
\end{equation}
Notice that for an encoding state $\rho_i$ and a simultaneous decoder $D_j=\sum_{x^n\in I_j}E_{x^n}$ with $I_j\subset\{0,1\}^n$,
\begin{equation}
\Tr\left(
\mathcal{N}_p^{\otimes n}(\rho_i)\, D_j
\right)=\Tr\left(
\mathcal{M}\circ\mathcal{N}_p^{\otimes n}(\rho_i)\,
\sum_{x^n \in I_j} \ket{x^n}\!\bra{x^n}
\right),
\label{eq: decision region}
\end{equation}
so simultaneous decoding is equivalent to classical decision regions. Recall the necessary condition~(\ref{sep req}) for the existence of an $(n,N,\lambda_1,\lambda_2)$-code:
\begin{equation}
\forall\, i \neq j,\qquad
\Tr\bigl[(\mathcal{N}_p^{\otimes n}(\rho_i) - \mathcal{N}_p^{\otimes n}(\rho_j))D_i\bigr] \ge 1 - \lambda_1 - \lambda_2.
\end{equation}
Since by~(\ref{eq: decision region}) a measurement in $\{D_j, \id - D_j\}$ can be implemented by first measuring $\mathcal{M}$ and then doing classical post-processing, for the same reason as~(\ref{sep req}), we have:
\begin{equation}
    \forall\, i \neq j,\qquad \frac12\bigl\| \mathcal{M}\circ\mathcal{N}_p^{\otimes n}(\rho_i)-\mathcal{M}\circ\mathcal{N}_p^{\otimes n}(\rho_j)\bigr\|_1 \ge 1 - \lambda_1 - \lambda_2.
    \label{sep req mmt}
\end{equation}

\paragraph{Reduction to $n$-fold BSC.}
By \zcref{lem:reduction-bsc}, for every input state $\rho$, the outcome distribution $\mathcal{M}\circ\mathcal{N}^{\otimes n}_p(\rho)$ depends only on the diagonal elements of the input $r_\rho(x^n):=\bra{\Psi_{x^n}}\rho\ket{\Psi_{x^n}}$, and is given by $q_\rho = \mathrm{BSC}_{p/2}^{\otimes n}(r_\rho)$.
In particular, we can rewrite the necessary condition~(\ref{sep req mmt}) as
\begin{equation}
\label{eq:TV_equals_trace_diag}
\forall i \neq j,\quad \bigl\| \mathrm{BSC}_{p/2}^{\otimes n}(r_{\rho_i})-\mathrm{BSC}_{p/2}^{\otimes n}(r_{\rho_j}))\bigr\|_{\mathrm{TV}}=\frac12\bigl\| \mathcal{M}\circ\mathcal{N}_p^{\otimes n}(\rho_i)-\mathcal{M}\circ\mathcal{N}_p^{\otimes n}(\rho_j)\bigr\|_1\ge 1 - \lambda_1 - \lambda_2,
\end{equation}
since the trace distance between two classical states is equal to their total variation distance.

\paragraph{Classical soft-covering.}
We now want to find an upper bound on the number of well-separated (in TV distance) output distributions for the $n$-fold BSC. Fix an arbitrary input distribution $P_{\uX}$ on $\{0,1\}^n$. Consider the state
\begin{equation}
    \rho_{X^nB^n}=\sum_{x^n}P_{\uX}(x^n)\ket{x^n}\!\bra{x^n}\otimes \rho_{B^n}^{x^n},
\qquad
\rho_{B^n}^{x^n}\coloneqq\sum_{y^n}\mathrm{BSC}_{p/2}^{\otimes n}(y^n|x^n)\ket{y^n}\!\bra{y^n},
\end{equation}
which is diagonal (hence classical) on $B^n$:
\begin{equation}
    \rho_{B^n}=\sum_{y^n}\sum_{x^n}P_{\uX}(x^n)\mathrm{BSC}_{p/2}^{\otimes n}(y^n|x^n)\ket{y^n}\!\bra{y^n}=\sum_{y^n}\mathrm{BSC}_{p/2}^{\otimes n}(P_{\uX})(y^n)\ket{y^n}\!\bra{y^n}.
\end{equation}
In fact, the bipartite state $\rho_{X^nB^n}$ is itself diagonal in the basis $\{\ket{x^n} \otimes \ket{y^n}\}_{x^n,y^n \in \{0,1\}^n}$, hence classical:
\begin{align} \rho_{X^nB^n} &= \sum_{x^n, y^n\in \{0,1\}^n} P_{\uX\uY}(x^n,y^n) \ket{x^n}\!\bra{x^n}
\otimes \ket{y^n}\!\bra{y^n},\end{align}
where for any $x^n,y^n \in \{0,1\}^n$,
\begin{align}
  P_{\uX\uY}(x^n,y^n) =P_{\uX}(x^n)\mathrm{BSC}_{p/2}^{\otimes n}(y^n|x^n),\label{eq:joint}
\end{align}
i.e.,$\uX$ and $\uY$ are random vectors with joint distribution given by (\ref{eq:joint}).

The soft-covering theorem, \zcref{thm:soft_covering}, {\em{with (\ref{eq:RHSbound}) replaced by (\ref{eq:clbound})}} since $\rho_{X^nB^n}$ is a purely classical state, implies that for any $\alpha\in(1,2)$ and any $M\in\mathbb{N}$, there exists a codebook $\mathcal{C}=\{x^n(1),\dots,x^n(M)\}$ where each codeword $x^n(i)\in\{0,1\}^n$, such that its induced average state
\begin{equation}
\rho_{B^n}^{\mathcal{C}}\coloneqq\frac1M\sum_{m=1}^M\rho_{B^n}^{x^n(m)}=\sum_{y^n\in \{0,1\}^n}\frac1M\sum_{m=1}^M \mathrm{BSC}_{p/2}^{\otimes n}\bigl(y^n| x^n(m)\bigr)\ket{y^n}\!\bra{y^n}\equiv\sum_{y^n\in \{0,1\}^n}P^\mathcal{C}_{Y^n}(y^n)\ket{y^n}\!\bra{y^n},
\end{equation}
where we have defined the codebook-induced output distribution: for any ${y\in \{0,1\}^n}$,
\begin{equation}
    P^\mathcal{C}_{Y^n}(y^n)\coloneqq \frac1M\sum_{m=1}^M \mathrm{BSC}_{p/2}^{\otimes n}\bigl(y^n| x^n(m)\bigr),
\end{equation}
satisfies
\begin{equation}
\label{eq:soft_covering_bound}
\bigl\|P_{Y^n}^{\mathcal{C}}-\mathrm{BSC}_{p/2}^{\otimes n}(P_{X^n})\bigr\|_{\mathrm{TV}}=\frac12\bigl\|\rho_{B^n}^{\mathcal{C}}-\rho_{B^n}\bigr\|_1
\le
2^{\frac{2}{\alpha}-2+
\frac{\alpha-1}{\alpha}
%\bigl(I_\alpha^*(X:B)_\rho - \log M \bigr)}.
\bigl(
I_\alpha^S({\uX}:{\uY})- \log M
\bigr)},
\end{equation}
where $I_\alpha^S({\uX}:{\uY})$ is the Sibson $\alpha$-mutual information
\begin{align}
   {I}_\alpha^S({\uX};{\uY}) &:= \inf_{Q_{\uY}}D_\alpha(P_{\uX\uY}||P_{\uX}Q_{\uY}),
   \end{align}
the infimum being over all probability distributions $Q_{\uY}$ of the random vector ${\uY}$.

\paragraph{Types.}
The codebook-induced output distribution $P^\mathcal{C}_{Y^n}$ can be equivalently viewed as the output of the $n$-fold BSC by inputting an $M$-type (i.e., an empirical distribution of $M$ samples). To see this, let us denote the empirical distribution of the codebook $\mathcal{C}$ by $\hat P_{\mathcal{C}}$: for any ${x^n\in \{0,1\}^n}$,
\begin{equation}
    \hat P_{\mathcal{C}}(x^n)\coloneqq\frac{1}{M}\sum_{m=1}^M \mathbf{1}\{x^n(m)=x^n\},
\end{equation}
so $\hat P_{\mathcal{C}}$ is by definition an $M$-type on the alphabet $\{0,1\}^n$. Indeed,
\begin{equation}
    \mathrm{BSC}_{p/2}^{\otimes n}(\hat P_{\mathcal{C}})(y^n)=\sum_{x^n} \hat P_{\mathcal{C}}(x^n)\mathrm{BSC}_{p/2}^{\otimes n}(y^n|x^n)=\frac{1}{M}\sum_{m=1}^M \sum_{x^n}\mathbf{1}\{x^n(m)=x^n\}\mathrm{BSC}_{p/2}^{\otimes n}(y^n|x^n) =P_Y^{\mathcal{C}}(y^n).
\end{equation}
Hence, \eqref{eq:soft_covering_bound} can be rephrased as follows: for every input distribution $P_{X^n}$ there exists
an $M$-type $\hat P_{\uX}$ such that
\begin{align}
\label{eq:soft_covering_type_form}
\forall \alpha\in(1,2), \quad \bigl\|\mathrm{BSC}_{p/2}^{\otimes n}(\hat P_{\uX})-\mathrm{BSC}_{p/2}^{\otimes n}(P_{\uX})\bigr\|_{\mathrm{TV}}
&\le
2^{\frac{2}{\alpha}-2+
\frac{\alpha-1}{\alpha}
\bigl(
I_\alpha^S(\uX:\uY)- \log M
\bigr)}.
\end{align}
 
\paragraph{Strong converse bound.}
Let $\varepsilon>0$, choose $M$ large enough so that
the right-hand side of \eqref{eq:soft_covering_type_form} is at most $\varepsilon$ for every possible input $P_{\uX}$.
A sufficient condition is, for any $\alpha \in (1,2)$, pick
\begin{equation}
\label{eq:M_choice_general}
M = \left\lceil2^{\sup_{P_{\uX}} I_\alpha^S(\uX; \uY) \;+\; \frac{\alpha}{\alpha-1}\left(\frac{2}{\alpha}-2-\log \varepsilon\right)}\right\rceil,
\end{equation}
where the supremum ranges over all input distributions on $\{0,1\}^n$. Under \eqref{eq:M_choice_general}, for each $i\in[N]$ there exists an $M$-type distribution $\hat P_i$
such that
\begin{equation}
\label{eq:each_i_approx}
\bigl\| \mathrm{BSC}_{p/2}^{\otimes n}(\hat{P}_i)-\mathrm{BSC}_{p/2}^{\otimes n}(r_{\rho_i}))\bigr\|_{\mathrm{TV}}\le \varepsilon.
\end{equation}
Recall that the $M$-type is an empirical distribution of a length-$M$ sequence over the alphabet $\{0,1\}^n$, hence the number of possible $M$-types is upper bounded by the number of sequences, that is $2^{nM}$. Assume that the number $N$ of messages in the identification code is larger than this upper bound: $N>2^{nM}$, then there must exist distinct indices $i\neq j$ such that $\hat P_i=\hat P_j$.
Denote this common $M$-type by $\hat P$. By the triangle inequality and \eqref{eq:each_i_approx},
\begin{equation}
    \bigl\| \mathrm{BSC}_{p/2}^{\otimes n}(r_{\rho_i})-\mathrm{BSC}_{p/2}^{\otimes n}(r_{\rho_j}))\bigr\|_{\mathrm{TV}}\le\bigl\| \mathrm{BSC}_{p/2}^{\otimes n}(r_{\rho_i})-\mathrm{BSC}_{p/2}^{\otimes n}(\hat{P}))\bigr\|_{\mathrm{TV}}+\bigl\| \mathrm{BSC}_{p/2}^{\otimes n}(r_{\rho_j})-\mathrm{BSC}_{p/2}^{\otimes n}(\hat{P}))\bigr\|_{\mathrm{TV}} \le2\varepsilon.
\end{equation}
For any $\varepsilon<\frac{1-\lambda_1-\lambda_2}{2}$, this contradicts the necessary condition \eqref{eq:TV_equals_trace_diag}, which implies
\begin{equation}
\label{eq: up bd N}
\tilde{N}(n,\lambda_1,\lambda_2) \le 2^{nM}\le 2^{n\left[2^{\sup_{P_{\uX}} I_\alpha^S(\uX:\uY)\;+\; \frac{\alpha}{\alpha-1}\left(\frac{2}{\alpha}-2-\log \varepsilon\right)}+1\right]}, \quad \forall \varepsilon\in\left(0,\frac{1-\lambda_1-\lambda_2}{2}\right), \alpha \in (1,2).
\end{equation}
Sibson $\alpha$-mutual information satisfies a single-letterization property: for a discrete memoryless channel $W^{\otimes n}(y^n|x^n)=\prod_i W(y_i|x_i)$, let $P_{XY}(x,y)=P_{X}(x)W(y|x)$ and $P_{\uX\uY}(x^n,y^n)=P_{\uX}(x^n)W^{\otimes n}(y^n|x^n)$, then (Theorem 6 in \cite{verdu2015alpha})
\begin{equation}
    \sup_{P_{\uX}} I_\alpha^S(\uX:\uY) = n \sup_{P_{X}} I_\alpha^S(X:Y).
\end{equation}
Set $W=\mathrm{BSC}_{p/2}$. Then, fixing an arbitrary $0<\varepsilon<\frac{1-\lambda_1-\lambda_2}{2}$, and an arbitrary $1<\alpha<2$ in (\ref{eq: up bd N}), and taking double logarithms on both sides of the inequality, yields
\begin{equation}
    \log\log \tilde{N}(n,\lambda_1,\lambda_2) \le \log n+\log\left[2^{n \sup_{P_{X}} I_{\alpha}^S(X:Y) }\;+\; \frac{\alpha}{\alpha-1}\left(\frac{2}{\alpha}-2-\log \varepsilon\right)+1\right]
\end{equation}
Then,
\begin{equation}
    \limsup_{n \to \infty} \frac1n\log\log \tilde{N}(n,\lambda_1,\lambda_2) \le \sup_{P_{X}} I_{\alpha}^S(X:Y), \quad \forall \alpha\in(1,2).
\end{equation}
Since the Sibson $\alpha$-mutual information is continuous in $\alpha$, and reduces to the mutual information for $\alpha\to1$ \cite{verdu2015alpha}, we can now take $\alpha\downarrow1$ to obtain
\begin{equation}
    \limsup_{n \to \infty} \frac1n\log\log \tilde{N}(n,\lambda_1,\lambda_2) \le \sup_{P_{X}} I(X:Y) =C_{\mathrm{Sh}}(\mathrm{BSC}_{p/2})
\end{equation}
where the last equality follows from Shannon's Noisy Channel Coding Theorem~\cite{shannon1948mathematical}, with $C_{\mathrm{Sh}}(\mathrm{BSC}_{p/2})$ denoting the classical Shannon capacity of the binary symmetric channel of crossover probability $p/2$. It is easily evaluated and seen to be given by $1-h(p/2)$, where $h(\cdot)$ denotes the binary entropy.
By \zcref{def: complete product ID} of the simultaneous identification capacity with complete product measurement, we have the strong converse bound: for any $\lambda_1, \lambda_2>0, \lambda_1+\lambda_2<1$,
\begin{equation}
    \tilde{C}_{\mathrm{ID}}^{\mathrm{sim}}(\mathcal{N}_p) \le \limsup_{n \to \infty}
\frac{1}{n} \log \log \tilde{N}_{\mathrm{sim}}(n,\lambda_1,\lambda_2)\le C_{\mathrm{Sh}}(\mathrm{BSC}_{p/2})=1-h(p/2) = C(\mathcal{N}_p),
\end{equation}
where $C(\mathcal{N}_p)$ denotes the classical capacity of the qubit depolarizing channel as defined above, and the last equality was proved by King in~\cite{king2003capacity}. In particular, there it was also shown that this capacity is achievable with product state inputs and complete product measurements.
Using the argument of \zcref{thm: achiev} this then also implies that there exists an identification code with complete product measurements that achieves $C_{\mathrm{ID}}(\mathcal{N}_p) = 1 - h(p/2)$. For this just note that the construction used in the proof of \zcref{thm: achiev} is purely classical - it employs a classical identification code for the classical identity channel, the latter being obtained via a transmission code for the noisy quantum channel. Thus, if the transmission code only uses complete product measurements then so does the identification code. Hence, we have matching upper and lower bounds for the simultaneous classical identification capacity using complete product measurements.

\end{proof}

\section{A strong converse bound for the unrestricted identification capacity}
\label{sec:unrestricted converse}
In this section, we establish a strong converse bound for the unrestricted identification capacity of the qubit depolarizing channel $\mathcal{N}_p$. Our strategy is to find a direct covering of the output geometry of the $n$-fold qubit depolarizing channel, and then use the necessary condition~(\ref{sep req}) to obtain the strong converse bound. This is possible due to the simplicity of the output geometry as a contraction of the Bloch sphere.

%%%%%%%%%%%%%%%%%%%%%%%%%%%%%%%%%%%%%%%%%%%%%%%%%%%%%%%%%%%%%%%%%%%%%%%%%%%%%%%%
\subsection[\texorpdfstring{Input- and output geometry of the $n$-fold qubit depolarizing channel $\mathcal{N}_p^{\otimes n}$}{Input and output geometry of the qubit depolarizing channel}]{Input and output geometry of the qubit depolarizing channel $\mathcal{N}_p^{\otimes n}$}

Let us start by representing an arbitrary $n$-qubit state in terms of tensor products of Pauli matrices. 
Let
\[
\mathcal{P}_n \coloneqq \big\{ \sigma_{\alpha_1} \otimes \cdots \otimes \sigma_{\alpha_n}
:\; \alpha_i \in \{0,1,2,3\} \big\},
\]
where $\sigma_0 = \id$ and $\sigma_1,\sigma_2,\sigma_3$ are the usual
Pauli matrices.  The set $\mathcal{P}_n$ contains $4^n$ operators and forms an orthogonal basis of $\mathcal{B}((\mathbb{C}^2)^{\otimes n})$ with respect to
the Hilbert--Schmidt inner product $\langle A,B\rangle \coloneqq \Tr(A^\dagger B)$. Define the normalized Pauli basis (from now on denote $d=2^n$)
\begin{equation}
\tilde{\sigma}_{\boldsymbol{\alpha}} :=
\frac{1}{\sqrt{d}}\, 
\sigma_{\alpha_1} \otimes \cdots \otimes \sigma_{\alpha_n},
\qquad \boldsymbol{\alpha} = (\alpha_1,\dots,\alpha_n) \in \{0,1,2,3\}^n,    
\end{equation}
which satisfies
\begin{equation}
    \langle \tilde{\sigma}_{\boldsymbol{\alpha}},\tilde{\sigma}_{\boldsymbol{\beta}}\rangle = \Tr\bigl(\tilde{\sigma}_{\boldsymbol{\alpha}}^\dagger
\tilde{\sigma}_{\boldsymbol{\beta}}\bigr)
= \delta_{\boldsymbol{\alpha}\boldsymbol{\beta}}.
\end{equation}
Then any $n$--qubit state $\rho \in \mathcal{D}(\mathbb{C}^d)$ admits the expansion
\begin{equation}
\rho = \sum_{\boldsymbol{\alpha} \in \{0,1,2,3\}^n}
r_{\boldsymbol{\alpha}}\, \tilde{\sigma}_{\boldsymbol{\alpha}}=\frac{\id_d}{d}+\sum_{\boldsymbol{\alpha}\neq (0,\dots,0)}
r_{\boldsymbol{\alpha}}\, \tilde{\sigma}_{\boldsymbol{\alpha}},
\end{equation}
with coefficients
\begin{equation}
r_{\boldsymbol{\alpha}} = \langle \rho,\tilde{\sigma}_{\boldsymbol{\alpha}}\rangle=
\Tr\bigl(\rho\, \tilde{\sigma}_{\boldsymbol{\alpha}}\bigr).
\end{equation}
Since $\Tr(\rho)=1$, we always have
$r_{(0,\dots,0)} = \frac{1}{\sqrt{d}}$, and since $\rho$ is Hermitian, all coefficients are real. That is, we can represent $\rho$ as a (Bloch) vector $r\in\mathbb{R}^D$, where from now on we denote $D=4^n-1$, with components $r_{\boldsymbol{\alpha}}$. For two states $\rho,\sigma \in \mathcal{D}(\mathbb{C}^d)$ we have
\begin{equation}
\left\|\rho-\sigma\right\|_2^2
= \Tr\left((\rho-\sigma)^2\right)
= \sum_{\boldsymbol{\alpha}} (r_{\boldsymbol{\alpha}} - s_{\boldsymbol{\alpha}})^2
= \left\|r-s\right\|_2^2. 
\label{eq: norm equiv}
\end{equation}
In particular, taking $\sigma = \id_d/d$ gives
\begin{equation}
\left\|r\right\|_2^2=\left\|\rho - \tfrac{\id_d}{d}\right\|_2^2
= \Tr(\rho^2) - \frac{1}{d}
\le 1 - \frac{1}{d},
\end{equation}
with equality if and only if $\rho$ is pure. Therefore the full $n$-qubit input space, when embedded in $\mathbb{R}^D$, is contained in a Euclidean ball of radius $\sqrt{1 - \tfrac{1}{d}}$ centered at origin:
\begin{equation}
\mathcal{D}(\mathbb{C}^d)
\subset B\!\left(0,\, \sqrt{1 - \tfrac{1}{d}}\right),
\end{equation}
where we have used the notation in~\zcref{def: epsilon covering}.

Given an input state $\rho$ with coordinates $r_{\boldsymbol{\alpha}}$, denote the coefficients of the output state under the $n$-fold depolarizing channel $\mathcal{N}_p^{\otimes n}$ by $r'_{\boldsymbol{\alpha}}$, then they are related by
\begin{equation}
\begin{aligned}
    r'_{\boldsymbol{\alpha}} &= \Tr\bigl( \mathcal{N}_p^{\otimes n}(\rho)\, \tilde{\sigma}_{\boldsymbol{\alpha}} \bigr)\\
    &=\Tr\Bigl[\Bigl( \frac{\id_d}{d}+\sum_{\boldsymbol{\beta}\neq (0,\dots,0)}
r_{\boldsymbol{\beta}}\, \frac{1}{\sqrt{d}}\, 
\mathcal{N}_p(\sigma_{\beta_1}) \otimes \cdots \otimes \mathcal{N}_p(\sigma_{\beta_n})\, \Bigl)\tilde{\sigma}_{\boldsymbol{\alpha}} \Bigr]\\
&=\Tr\Bigl[\Bigl( \frac{\id_d}{d}+\sum_{\boldsymbol{\beta}\neq (0,\dots,0)}
r_{\boldsymbol{\beta}}\, (1-p)^{w(\boldsymbol{\beta})}\frac{1}{\sqrt{d}}\, 
\sigma_{\beta_1} \otimes \cdots \otimes \sigma_{\beta_n}\, \Bigl)\tilde{\sigma}_{\boldsymbol{\alpha}} \Bigr]\\
&=\sum_{\boldsymbol{\beta}\neq (0,\dots,0)}
r_{\boldsymbol{\beta}}\, (1-p)^{w(\boldsymbol{\beta})}\Tr\bigl(\tilde{\sigma}_{\boldsymbol{\beta}}\, \tilde{\sigma}_{\boldsymbol{\alpha}}\bigl)\\
&=(1-p)^{w(\boldsymbol{\alpha})} r_{\boldsymbol{\alpha}},
\end{aligned}
\end{equation}
where the weight $w(\boldsymbol{\alpha})$ denotes the number of non-zero letters in the string $\boldsymbol{\alpha} = (\alpha_1,\dots,\alpha_n)$, with $w(\boldsymbol{\alpha}) \in \{1,\cdots,n\}$.
Thus the Pauli coefficients transform as a non-uniform contraction:
\begin{equation}
  r_{\boldsymbol{\alpha}} \longmapsto r'_{\boldsymbol{\alpha}} = (1-p)^{w(\boldsymbol{\alpha})} r_{\boldsymbol{\alpha}}.
\end{equation}
Since the input space is contained in a Euclidean ball, the output space will be contained in an Euclidean ellipsoid $\mathcal{S}^{'}_n$:
\begin{equation}
    \mathcal{N}_p^{\otimes n}\left(\mathcal{D}(\mathbb{C}^d)\right)
\subset \mathcal{N}_p^{\otimes n}\left(B\!\left(0,\, \sqrt{1 - \tfrac{1}{d}}\right)\right) \equiv \mathcal{S}^{'}_n \subset \mathbb{R}^D,
\label{eq: output geometry 1}
\end{equation}
with semi-axes $\{a^{'}_{\boldsymbol{\alpha}}:\ \boldsymbol{\alpha} \in \{0,1,2,3\}^n, \boldsymbol{\alpha} \neq (0,\dots,0)\}$ given by
\begin{equation}
    a^{'}_{\boldsymbol{\alpha}} = (1-p)^{w(\boldsymbol{\alpha})}\sqrt{1 - \tfrac{1}{d}}.
    \label{eq: output geometry 2}
\end{equation}

%%%%%%%%%%%%%%%%%%%%%%%%%%%%%%%%%%%%%%%%%%%%%%%%%%%%%%%%%%%%%%%%%%%%%%%%%%%%%%%%
\subsection{Ellipsoid covering and identification converse}
\begin{definition}
[(Euclidean $\varepsilon$--covering)]
\label{def: epsilon covering}
Let $y = (y_1,\dots,y_n) \in \mathbb{R}^D$. The Euclidean ball of
radius $\varepsilon > 0$ centered at $y$ is denoted as\footnote{We refer to such a ball as an $\varepsilon$-ball. For $\varepsilon =1$, the corresponding ball is referred to as a $1$-ball.}
\begin{equation}
  B(y,\varepsilon)
  := \bigl\{\, x \in \mathbb{R}^D \;\big|\; \sum_{i=1}^D (x_i - y_i)^2 \le \varepsilon^2 \,\bigr\}.
\end{equation}
For a subset $A \subset \mathbb{R}^D$, a subset $\mathcal{M}_\varepsilon(A) \subset \mathbb{R}^D$
is called an $\varepsilon$--covering of $A$ if
\begin{equation}
  A \subset \bigcup_{y \in \mathcal{M}_\epsilon(A)} B(y,\varepsilon).
\end{equation}
The minimum cardinality over all $\varepsilon$-coverings of $A$ is denoted as $|\mathcal{M}^*_\varepsilon(A)|$.
\end{definition}
\begin{lemma}\label{lem:covering and ID}
Let $D=4^n-1, d=2^n$, and let  $N(n,\lambda_1,\lambda_2)$ be the maximal size of any $(n,N,\lambda_1,\lambda_2)$ classical identification code for the qubit depolarizing channel $\mathcal{N}_p$. Let $\mathcal{S}_n \in \mathbb{R}^D$ be an ellipsoid with semi-axes given by 
\begin{equation}
    \left\{a_{\boldsymbol{\alpha}}= (1-p)^{w(\boldsymbol{\alpha})}\frac{\sqrt{d - 1}}{1-\lambda_1-\lambda_2}:\ \boldsymbol{\alpha} \in \{0,1,2,3\}^n, \boldsymbol{\alpha} \neq (0,\dots,0)\right\}.
\end{equation}
Then for all $n\in \mathbb{Z}_+, \lambda_1>0, \lambda_2>0, \lambda_1+\lambda_2<1$,
\begin{equation}
  N(n,\lambda_1,\lambda_2) \le \bigl|M^*_1(\mathcal{S}_n)\bigr|.
\end{equation}
\end{lemma}

\begin{proof}
First note that given a subset $A\in\mathbb{R}^D$ and an $\varepsilon$-covering $\mathcal{M}_\varepsilon(A)$, we can always replace the $\varepsilon$-balls with $1$-balls to obtain a covering of the rescaled subset $A/\varepsilon$. This implies that the minimum $\varepsilon$-covering of $A$ is equivalent to the minimum $1$-covering of $A/\varepsilon$: $\bigl|M^*_1(A/\varepsilon)\bigr|=\bigl|M^*_\varepsilon(A)\bigr|$. Therefore,
\begin{equation}
    \bigl|M^*_1(\mathcal{S}_n)\bigr|=\Bigl|M^*_{\frac{1-\lambda_1-\lambda_2}{\sqrt{d}}}(\mathcal{S}^{'}_n)\Bigr|,
\end{equation}
where $\mathcal{S}^{'}_n$ is the ellipsoid with semi-axes given by 
\begin{equation}
    \left\{a^{'}_{\boldsymbol{\alpha}}= (1-p)^{w(\boldsymbol{\alpha})}\sqrt{1 - \tfrac{1}{d}}:\ \boldsymbol{\alpha} \in \{0,1,2,3\}^n, \boldsymbol{\alpha} \neq (0,\dots,0)\right\}.
\end{equation}
By (\ref{eq: output geometry 1}) and (\ref{eq: output geometry 2}), $\mathcal{S}^{'}_n=\mathcal{N}_p^{\otimes n}\left(B\!\left(0,\, \sqrt{1 - \tfrac{1}{d}}\right)\right) \supset\mathcal{N}_p^{\otimes n}\left(\mathcal{D}(\mathbb{C}^d)\right)$.

Let $\bigl\{ (\rho_i, D_i) : i = 1,\dots, N \bigr\}$ be any $(n,N,\lambda_1,\lambda_2)$ unrestricted identification code for $\mathcal{N}_p$. Recall the necessary condition~(\ref{sep req}) for the existence of an identification code:
\begin{equation}
\forall\, i \neq j,\qquad
\frac{1}{2}\left\| \mathcal{N}_p^{\otimes n}(\rho_i) - \mathcal{N}_p^{\otimes n}(\rho_j) \right\|_1 \ge 1 - \lambda_1 - \lambda_2.
\end{equation}
Using $\|X\|_1 \le \sqrt{d}\|X\|_2$ for $X \in \mathcal{B}(\mathbb{C}^d)$, (\ref{eq: norm equiv}) implies the necessary condition
\begin{equation}
  \left\|r_i - r_j\right\|_2 = \left\|\mathcal{N}_p^{\otimes n}(\rho_i) - \mathcal{N}_p^{\otimes n}(\rho_j)\right\|_2
  \ge \frac{2(1-\lambda_1-\lambda_2)}{\sqrt{d}},
  \label{eq: necessary}
\end{equation}
where $r_i$ is the (Bloch) vector corresponding to the state $\mathcal{N}_p^{\otimes n}(\rho_i)$.

We prove by contradiction. Consider the minimal $\varepsilon$-covering $M^*_{\varepsilon}(\mathcal{S}^{'}_n)$ of $\mathcal{S}^{'}_n$, and assume \begin{equation}
    N(n,\lambda_1,\lambda_2)>\Bigl|M^*_{\varepsilon}(\mathcal{S}^{'}_n)\Bigr|.
\end{equation}
Since $r_i\in \mathcal{S}^{'}_n, \ \forall i=1,\cdots,N$, by definition of the $\varepsilon$-covering, there exists $y^i\in M^*_{\varepsilon}(\mathcal{S}^{'}_n)$, such that
\begin{equation}
    \left\|r_i - y^i\right\|_2 \le \varepsilon.
\end{equation}
But since $N(n,\lambda_1,\lambda_2)>\Bigl|M^*_{\varepsilon}(\mathcal{S}^{'}_n)\Bigr|$, there must exist two distinct codewords $\rho_i \neq \rho_j$ with $y^i=y^j$, that is, the (Bloch) vectors corresponding to their outputs lie in the same $\varepsilon$-ball:
\begin{equation}
  \exists \ i\neq j,\quad \left\|r_i - r_j\right\|_2 \leq \left\|r_i - y^i\right\|_2 +\left\|r_j - y^j\right\|_2  \leq 2\varepsilon.
\end{equation}
For any $\varepsilon\in(0,\frac{1-\lambda_1-\lambda_2}{\sqrt{d}})$, this leads to a contradiction with the necessary condition~(\ref{eq: necessary}). Therefore,
\begin{equation}
    N(n,\lambda_1,\lambda_2)\le\Bigl|M^*_{\varepsilon}(\mathcal{S}^{'}_n)\Bigr|, \quad, \forall \varepsilon\in\left(0,\frac{1-\lambda_1-\lambda_2}{\sqrt{d}}\right).
\end{equation}
Taking the limit $\varepsilon \to \frac{1-\lambda_1-\lambda_2}{\sqrt{d}}$ completes the proof.
\end{proof}

Our task of establishing a strong converse bound for the unrestricted identification capacity for $\mathcal{N}_p$ now reduces to studying the minimal covering number $\bigl|M^*_1(\mathcal{S}_n)\bigr|$ in~\zcref{lem:covering and ID}. For this, we use the following theorem from \cite{dumer2006covering}.
\begin{boxed}
\begin{theorem}[(Theorem 2 in \cite{dumer2006covering})]\label{thm:ellipsoid covering}
Let $\mathcal{S}_n \subset \mathbb{R}^n$ be an ellipsoid with semi-axes $\{a_1,\cdots,a_n\}$:
\begin{equation}
\mathcal{S}_n
\coloneqq \left\{ x \in \mathbb{R}^n \;\middle|\;
\sum_{i=1}^n \frac{x_i^2}{a_i^2} \le 1 \right\}.
\end{equation}

For any $\theta \in (0,1/2)$, define the index sets
\begin{align}
J_{\theta,0} & \coloneqq \{\, i : a_i^2 \le 1-\theta \,\},
\label{eq: partition 1}\\
J_{\theta,1} & \coloneqq \{\, i : 1-\theta < a_i^2 \le 1 \,\},
\label{eq: partition 2}\\
J & \coloneqq \{\, i : a_i > 1 \,\}.
\label{eq: partition 3}
\end{align}
Define
\begin{equation}
\mu_\theta \coloneqq |J_{\theta,1}| + |J|, \quad K \coloneqq \sum_{i \in J} \ln a_i .
\end{equation}

Then the minimal 1-ball covering of the ellipsoid $\mathcal{S}_n$ satisfies
\begin{equation}
\ln |M^*_1(\mathcal{S}_n)|
\le
K + \mu_\theta \ln\!\left(\frac{3}{\theta}\right).
\end{equation}
\end{theorem}
\end{boxed}

We are now ready to state our main result, which is an application of \zcref{thm:ellipsoid covering} to the ellipsoid in~\zcref{lem:covering and ID}.

\begin{boxed}
\begin{theorem}[(Strong converse bound for the unrestricted classical identification capacity of $\mathcal{N}_p$)]\label{thm: depolarizing strong converse}
The unrestricted identification capacity of the qubit depolarizing channel with error probability $p\in[0,1]$ satisfies the following strong converse bound: for any $\lambda_1,\lambda_2 >0$, such that $\lambda_1+\lambda_2 <1$,
\begin{equation}
    C_{\mathrm{ID}}(\mathcal{N}_p) \le \limsup_{n \to \infty}
\frac{1}{n} \log \log N(n,\lambda_1,\lambda_2) \leq 
\begin{cases}
2,
& 0 \le p \le 1 - 2^{-2/3}, \\[0.6em]
2 - D\!\left(\gamma(p)\,\middle\|\,\frac{3}{4}\right),
& 1 - 2^{-2/3} \le p < 1,
\end{cases}
\label{Thm 5 result}
\end{equation}
where 
\begin{equation}
\gamma(p) \coloneqq \frac{-1}{2 \log(1-p)},
\end{equation}
and the binary relative entropy is defined as
\begin{equation}
D(x\|y)
\coloneqq x \log\frac{x}{y}
+ (1-x)\log\frac{1-x}{1-y}.
\end{equation}

\textbf{Remark:} (\ref{Thm 5 result}) is of course also a strong converse bound for the simultaneous identification capacity.
\end{theorem}
\end{boxed}

\begin{proof}
From \zcref{lem:covering and ID}, we already know that
\begin{equation}
  N(n,\lambda_1,\lambda_2) \le \bigl|M^*_1(\mathcal{S}_n)\bigr|,
\end{equation}
with $\mathcal{S}_n \in \mathbb{R}^D$ ($D=4^n-1$) an ellipsoid with semi-axes given by 
\begin{equation}
    \left\{a_{\boldsymbol{\alpha}}= (1-p)^{w(\boldsymbol{\alpha})}\frac{\sqrt{2^n - 1}}{1-\lambda_1-\lambda_2}:\ \boldsymbol{\alpha} \in \{0,1,2,3\}^n, \boldsymbol{\alpha} \neq (0,\dots,0)\right\}.
\end{equation}
Define
\begin{equation}
C_n \coloneqq \frac{\sqrt{2^n-1}}{1-\lambda_1-\lambda_2},
\end{equation}
so for fixed $n$, the semi-axis in direction $\boldsymbol{\alpha}$ is solely dependent on its weight $w(\boldsymbol{\alpha})\in\{1,\cdots, n\}$, the number of non-zero letters in the string:
\begin{equation}
w(\boldsymbol{\alpha}) \coloneqq \big|
\left\{
i \in \{1,\dots,n\}
:\;
\alpha_i \neq 0
\right\}
\big|.
\end{equation}

\paragraph{Weight partition.}
For $k=1,\dots,n$, let
\begin{equation}
N_k
\coloneqq
\big|
\left\{
\boldsymbol{\alpha} \in \{0,1,2,3\}^n
:\;
w(\boldsymbol{\alpha}) = k
\right\}
\big|=
\binom{n}{k}3^k,
\end{equation}
the number of directions with weight $k$. One can check $\sum_{k=1}^n N_k = 4^n-1 = D$. All directions with weight $k$ share the semi-axis length
\begin{equation}
a(k) := C_n(1-p)^k.
\end{equation}
Fix $\theta\in(0,1/2)$ and define
\begin{equation}
\begin{cases}
k_{\theta}(n) \coloneqq
\max\left\{
k\in\{1,\dots,n\} :
a^2(k)=C_n^2(1-p)^{2k} > 1-\theta
\right\}, \\[0.6em]
k_{0}(n) \coloneqq
\max\left\{
k\in\{1,\dots,n\} :
a(k)=C_n(1-p)^{k} > 1
\right\}.
\end{cases}
\end{equation}
This partition of $k$'s corresponds to the partition of axes in~(\ref{eq: partition 1}-\ref{eq: partition 3}), hence
\begin{equation}
\mu_\theta
=
\sum_{k=1}^{k_\theta(n)} N_k,
\qquad
K
=
\sum_{k=1}^{k_0(n)} N_k \ln a(k).
\end{equation}
By \zcref{thm:ellipsoid covering}, we have
\begin{equation}
    \ln |M^*_1(\mathcal{S}_n)|
\le
K + \mu_\theta \ln\!\left(\frac{3}{\theta}\right).
\end{equation}
Notice that $\ln a(k)=\ln C_n + k\ln(1-p)\le \ln C_n$, and that $k_\theta(n)\ge k_0(n)$, we can loosen the bound:
\begin{equation}
\ln |M^*_1(\mathcal{S}_n)| \le
(\ln C_n) \sum_{k=1}^{k_\theta(n)} N_k+ \mu_\theta \ln\!\left(\frac{3}{\theta}\right)=\mu_\theta \ln\left(\frac{3C_n}{\theta}\right).
\label{eq: upbd ellip}
\end{equation}

\paragraph{Upper-bounding $\mu_\theta$.}
Let $\boldsymbol{\alpha}$ be uniformly random drawn from $\{0,1,2,3\}^n$, then
\begin{equation}
\Pr\big(w(\boldsymbol{\alpha})=k\big)
=
\frac{N_k}{4^n}
=
\binom{n}{k}
\left(\frac{3}{4}\right)^k
\left(\frac14\right)^{n-k}, \quad \Rightarrow \quad w(\boldsymbol{\alpha})
\sim
\mathrm{Bin}\!\left(n,\frac34\right).
\end{equation}
Therefore,
\begin{equation}
\mu_\theta
= 4^n\sum_{k=1}^{k_\theta(n)} N_k/4^n =
4^n \Pr\big(1\le w(\boldsymbol{\alpha})\le k_\theta(n)\big) \le 4^n \Pr\big(w(\boldsymbol{\alpha})\le k_\theta(n)\big).
\end{equation}
We can solve for $k_\theta(n)$ explicitly from its definition:
\begin{equation}
C_n^2(1-p)^{2k_\theta(n)} > 1-\theta
\quad\Longrightarrow\quad
k_\theta(n)
<
\frac{\log(1-\theta)-2\log C_n}{2\log(1-p)}.
\end{equation}
Plugging back, we have
\begin{equation}
\mu_\theta \le 4^n \Pr\Big(w(\boldsymbol{\alpha})\le \frac{\log(1-\theta)-2\log C_n}{2\log(1-p)}\Big).
\end{equation}

\paragraph{Chernoff-Cramér bound.}
We want to bound the left-tail probability $\Pr\big(X\le an+b\big)$ for some constant $a>0, b\ge0$, where $X\sim\mathrm{Bin}\!\left(n,q\right)$. For any $t\ge 0$,
\begin{equation}
\begin{aligned}
\Pr(X \le an+b)
\le
\Pr\left(2^{-tX} \ge 2^{-t(an+b)}\right)
&\le
2^{t(an+b)}\, \mathbb{E}\!\left[2^{-tX}\right] \nonumber\\
&=
2^{t(an+b)}\, \left(q2^{-t}+1-q\right)^n \nonumber\\
&=
2^{n\Bigl[ta+\log\left(q2^{-t}+1-q\right)\Bigr]+tb}
\eqqcolon f(t).
\end{aligned}
\end{equation}
Minimizing the exponent over $t\ge 0$, the optimizer is found to be
\begin{equation}
t^*
=
\begin{cases}
\log\frac{q(1-a-b/n)}{(a+b/n)(1-q)}, & a+\frac{b}{n}\le q,\\
0, & a+\frac{b}{n} \ge q,
\end{cases}
\end{equation}
so that
\begin{equation}
\inf_{t\ge 0} f(t)
=
\begin{cases}
2^{-nD\left(a+\frac{b}{n}\big|\big|q\right)}, & a+\frac{b}{n}\le q,\\
1, & a+\frac{b}{n} \ge q.
\end{cases}
\end{equation}
Putting together, 
\begin{equation}
\Pr(X \le an+b)
\le \begin{cases}
2^{-nD\left(a+\frac{b}{n}\big|\big|q\right)}, & a+\frac{b}{n}\le q,\\
1, & a+\frac{b}{n} \ge q.
\end{cases}
\end{equation}

\paragraph{Asymptotics.}
Recall that we want an upper bound for
\begin{equation}
\mu_\theta \le 4^n \Pr\Big(w(\boldsymbol{\alpha})\le \frac{\log(1-\theta)-2\log C_n}{2\log(1-p)}\Big), \quad 
C_n = \frac{\sqrt{2^n-1}}{1-\lambda_1-\lambda_2}.
\end{equation}
For large $n$,
\begin{equation}
    \frac{\log(1-\theta)-2\log C_n}{2\log(1-p)}\sim \gamma(p)n+ c,
\end{equation}
where $c$ is a constant and can be chosen to be non-negative without loss of generality, because if it is negative, we can always replace it by $0$ to obtain an upper bound on the probability.
\begin{enumerate}
\item If $\gamma(p)<\frac34$, corresponding to $1 - 2^{-2/3}< p< 1$, then for all sufficiently large $n$ we have
\begin{equation}
\gamma(p)+\frac{c}{n}\le \frac34,    
\end{equation}
therefore
\begin{equation}
\mu_\theta
\le
4^n 2^{-nD\left(\gamma(p)+\frac{c}{n}\big|\big|\frac34\right)}=
2^{\left[2-D\left(\gamma(p)+\frac{c}{n}\big|\big|\frac34\right)\right]n}.
\end{equation}
This together with \zcref{lem:covering and ID} and (\ref{eq: upbd ellip}), yields
\begin{equation}
\begin{aligned}
    \frac{1}{n} \log \log N(n,\lambda_1,\lambda_2)&\le\frac{1}{n}\log\log |M^*_1(\mathcal{S}_n)|\\
    &\le \frac{1}{n}\log\mu_\theta +\frac{1}{n}\log\log\left(\frac{3C_n}{\theta}\right)\\
    &\le 2-D\left(\gamma(p)+\frac{c}{n}\Big|\Big|\frac34\right)+O\left(\frac{\log n}{n}\right),
\end{aligned}
\end{equation}
taking $n\to\infty$, we obtain $C_{\mathrm{ID}}(\mathcal{N}_p) \le 2 - D\!\left(\gamma(p)\,\middle\|\,\frac{3}{4}\right)$.
\item If $\gamma(p)\ge\frac34$, corresponding to $0\le p\le 1 - 2^{-2/3}$, then the trivial bound gives
\begin{equation}
\mu_\theta \le 4^n = 2^{2n}.
\end{equation}
Combining this with \zcref{lem:covering and ID} and (\ref{eq: upbd ellip}), we have
\begin{equation}
\begin{aligned}
    \frac{1}{n} \log \log N(n,\lambda_1,\lambda_2)&\le\frac{1}{n}\log\log |M^*_1(\mathcal{S}_n)|\\
    &\le \frac{1}{n}\log\mu_\theta +\frac{1}{n}\log\log\left(\frac{3C_n}{\theta}\right)\\
    &\le 2+O\left(\frac{\log n}{n}\right),
\end{aligned}
\end{equation}
taking $n\to\infty$, we obtain $C_{\mathrm{ID}}(\mathcal{N}_p) \le 2$.
\end{enumerate}

\end{proof}

\textbf{Remark:} Take the limit $p\to1$ in~\zcref{thm: depolarizing strong converse}, we have $\gamma(p)\to0$ and $D(\gamma(p)||3/4)\to2$, hence the 
strong converse bound approaches zero. This is a desired feature of the strong converse bound that previous results fail to exhibit. As pointed out in~\zcref{sec:summary}, those bounds always include the logarithm of the dimension of the input or output Hilbert space, thereby preventing the bound on the capacity from decaying to zero as the channel becomes increasingly noisy.

\section{For general channels: A strong converse bound for the unrestricted identification capacity in terms of the classical capacity}
\label{sec:general converse}
\renewcommand{\N}{\mathcal{N}}

Up to this point, we have focused exclusively on the qubit depolarizing channel, owing to the simplicity of its output geometry, which corresponds to a contraction of the Bloch sphere. In this section, we establish a strong converse bound for arbitrary quantum channels, without imposing the simultaneity constraint. As noted in~\zcref{sec:summary}, we do not presently know how to eliminate the logarithmic dependence on the input dimension. Nevertheless, we show that the strong converse quantum capacity appearing in~(\ref{Q-soft-covering bound}) can be replaced by the classical capacity. This yields an explicit characterization for channels whose classical capacities are given by single-letter expressions, for example, the qubit depolarizing channel.

We start by stating a soft-covering result for fully quantum channels, which aims at approximating any fixed output state of a channel with the image of an input state of low rank.

\begin{boxed}
\begin{lemma}[(Quantum soft-covering with fixed input)] \label{thm:fixed-input-soft-covering}
    Let $\N: A \to B$ be a quantum channel with $\dim \mathcal{H}_A = d$, and $ρ \in \mathcal{D}(A)$ be an input quantum state, and $α \in (1,2)$. Consider the eigenvalue decomposition $ρ = \sum_{x = 1}^d p(x) \ketbra{\phi_x}$ and set $ρ_{XB} = \sum_{x = 1}^d p(x) \ketbra{x} \otimes \N(\ketbra{\phi_x})$. We write $\mathcal{P}_M([d])$ for the set of $M$-types ($M\in\mathbb{Z_+}$) on $[d] = \{1, ..., d\}$. For any such $M$-type $q \in \mathcal{P}_M([d])$, define $σ_q \coloneqq \sum_{x = 1}^d q(x) \ketbra{\phi_x}$. Then, for any $\varepsilon>0$, the condition
    \begin{equation}
    M \geq 2^{\tilde{I}_α(X:B)_{ρ_{XB}} - {α \over α- 1} \logε - 2}
    \end{equation}
    implies
    \begin{equation}
    \inf_{q \in \mathcal{P}_M([d])} {1 \over 2} \norm{\N(ρ) - \N(σ_q)}{1} \leq ε.
    \end{equation}
\end{lemma}
\end{boxed}

\begin{proof}
    This is a fairly direct application of \zcref{thm:soft_covering}. Note that sampling $M$ times from the probability distribution $\{p(x)\}_{x=1}^d$ yields a string $\underline{x} = x_1\ldots x_M \in [d]^{\times M}$ whose empirical probability distribution is an $M$-type $q_{\underline{x}} \in \mathcal{P}_M([d])$, i.e., $q_{\underline{x}}(x)=\frac{1}{M}\sum_{i=1}^M \mathbf{1}\{x_i=x\}$. 
This yields
    \begin{equation}
        \N(σ_{q_{\underline{x}}}) = {1 \over M} \sum_{i = 1}^M \N(\ketbra{\phi_{x_i}}).
    \end{equation}  
    \zcref{thm:soft_covering} then states that (note that with $ρ_{XB}$ defined above $ρ_B =\N(ρ)$),
    \begin{equation}
        {1 \over 2} \mathbb{E}_{\underline{x}}  \norm{\N(σ_{q_{\underline{x}}}) - \N(ρ)}{1} \leq 2^{{2 \over α} - 2}2^{{α - 1 \over α}[\tilde{I}_α(X:B)_{ρ_{XB}} - \log M]},
    \end{equation}
    and hence
    \begin{equation}
       {1 \over 2} \inf_{q \in \mathcal{P}_M([d])} \norm{\N(σ_{q}) - \N(ρ)}{1} \leq 2^{{2 \over α} - 2}2^{{α - 1 \over α}[\tilde{I}_α(X:B)_{ρ_{XB}} - \log M]}.
    \end{equation}
    Then the condition 
    \begin{equation}
        2^{{2 \over α} - 2}2^{{α - 1 \over α}[\tilde{I}_α(X:B)_{ρ_{XB}} - \log M]} \leq ε,
    \end{equation} 
    leads to the desired bound on $M$.
\end{proof}

\begin{corollary}\label{cor:low-rank-covering}
    For any $ρ \in \mathcal{D}(A)$ and any quantum channel $\N: A \to B$, there exists a state $σ \in \mathcal{D}(A)$ with $\mathrm{rank}(σ) \leq \ceil{2^{\tilde{I}_α(X:B)_{\rho_{XB}} - {α \over α- 1}\logε - 2}}$ with $\rho_{XB}$ defined in \zcref{thm:fixed-input-soft-covering}, such that ${1 \over 2}\norm{\N(ρ)- \N(σ)}{1} \leq ε$.
\end{corollary}

\begin{proof}
    This follows immediately from \zcref{thm:fixed-input-soft-covering}, by noting that for any $M$-type $q$, $q$ has at most $M$ non-zero entries, and hence the rank of $σ_q$ is at most $M$. 
\end{proof}

\begin{boxed}
\begin{theorem}\label{thm: general covering bound}
For all $λ_1,λ_2 > 0$ such that $λ_1 + λ_2 < 1$, we have the following strong converse bound on classical identification over an arbitrary quantum channel $\N: A \to B$:
\begin{equation}
C_{\mathrm{ID}}(\N)\le\limsup_{n \to \infty} {1 \over n}\log \log N(n, λ_1, λ_2) \leq \log |A| + C(\N),
\end{equation}
where $|A|$ is the dimension of the channel input space, and $C(\N)$ is the classical capacity.
\end{theorem}
\end{boxed}

\begin{proof}
    Let an identification code $\{(\rho_i,D_i): i=1,\dots,N\}$ for $\N^{\otimes n}$ with errors $λ_1$, $λ_2$ be given. Fix an $0<\varepsilon<\frac{1-\lambda_1-\lambda_2}{2}$. From \zcref{cor:low-rank-covering} we get that for every $i =1, ..., N$, there exists a state $σ_i$ with 
    \begin{equation}
        \mathrm{rank}(σ_i) \leq 2^{\tilde{I}_α(X^n:B^n)_{ω_i} - {α \over α- 1}\logε - 2} + 1,
    \end{equation}
    where $\omega_i = \sum_{x^n} p(x^n) \ketbra{x^n} \otimes \N(\ketbra{\phi_{x^n}})$ is a CQ state, with $\sum_{x^n} p(x^n)\ketbra{\phi_{x^n}}$ the eigenvalue decomposition of $\rho_i$, such that
    \begin{equation}
        {1 \over 2}\norm{\N^{\otimes n}(ρ_i) - \N^{\otimes n}(σ_i)}{1} \leq ε.
    \end{equation}
    Hence, by the variational characterization of the trace distance, for every $j = 1, ..., N$,
    \begin{equation}
    \left|\Tr[\N^{\otimes n}(σ_i) D_j] - \Tr[\N^{\otimes n}(ρ_i)D_j]\right| \leq {1 \over 2} \norm{\N^{\otimes n}(ρ_i) - \N^{\otimes n}(σ_i)}{1} \leq ε,    
    \end{equation}
    and so $\{(\sigma_i,D_i): i=1,\dots,N\}$ is an $(n, λ_1 + ε, λ_2 + ε)$ identification code for $\N^{\otimes n}$. In particular, this necessitates that
    \begin{equation}
        {1 \over 2} \norm{σ_i - σ_j}{1} \geq {1 \over 2} \norm{\N^{\otimes n}(σ_i)- \N^{\otimes n}(σ_j)}{1} \geq 1 - λ_1 - λ_2 - 2ε.
    \end{equation}
    Since for all $i$,
    \begin{equation}
        \mathrm{rank}(σ_i) \leq \max_{j = 1,..., N} 2^{\tilde{I}_α(X^n:B^n)_{ω_j} - {α \over α- 1}\logε - 2} + 1 \leq 2^{\sup_{ω} \tilde{I}_α(X^n:B^n)_{ω} - {α \over α- 1}\logε - 2} + 1,
    \end{equation}
    where we have relaxed the bound by taking the supremum over all CQ states of the form $\omega_{X^nB^n}=\sum_{x^n} \mu(x^n) \ketbra{x^n} \otimes \N(\nu_{x^n})$, where $\{\mu(x^n), \nu_{x^n}\}$ denotes an arbitrary ensemble of quantum states in ${\mathcal{D}}(A^n)$. There exists a reference system with Hilbert space $\mathcal{H}_R$ with dimension 
    \begin{equation}
        |R| = 2^{\sup_{ω} \tilde{I}_α(X^n:B^n)_{ω} - {α \over α- 1}\logε - 2} +~1,
    \end{equation}
    such that each $σ_i \in \mathcal{H}_A^{\otimes n}$ has a purification $\ket{\psi_i} \in \mathcal{H}_A^{\otimes n} \otimes \mathcal{H}_R$, and then by the data-processing inequality, using the notation $\psi=\ketbra{\psi}$,
    \begin{equation}
        {1 \over 2} \norm{\psi_i - \psi_j}{1} \geq 1 - λ_1 - λ_2 - 2 ε.
    \end{equation}

    It has been shown (Lemma 4 in \cite{bennett2005remote}) that for every $\delta>0$, every Hilbert space $\mathcal{H}$ of dimension $d$ has a $δ$-net of cardinality at most $(5/δ)^{2d}$, i.e., there exists a set of at most $(5/δ)^{2d}$ many pure states such that any pure state in $\mathcal{H}$ is at most $δ$ (in trace distance) away from one of the states in the set. Thus, in particular, any set of pure states $\{\ket{\psi_i} \in \mathcal{H}\}_i$ that satisfies ${1 \over 2} \norm{\psi_i - \psi_j}{1} > 2δ$ for all $i \neq j$ can be at most of cardinality $(5/δ)^{2d}$, since otherwise two states would have to be within  a distance $δ$ of one element from the $δ$-net, which implies that they are at most a distance $2δ$ away from each other by the triangle inequality. Hence, we find that
    \begin{equation}
        N(n,\lambda_1,\lambda_2) \leq \left({10 \over 1 - λ_1 - λ_2-2\varepsilon}\right)^{2 |A|^n |R|},
    \end{equation}
    which implies that 
    \begin{equation}
        \log \log N(n,\lambda_1,\lambda_2) \leq n \log |A| + \sup_\omega\tilde{I}_α(X^n:B^n)_{ω} + O(1).
    \end{equation}
    Dividing by $n$ and taking the limit $n\to\infty$, we have
    \begin{equation}
        \limsup_{n\to\infty}\frac1n\log \log N(n,\lambda_1,\lambda_2) \leq \log |A| + \lim_{n\to\infty}\frac1n\sup_\omega \tilde{I}_α(X^n:B^n)_{ω}.
    \end{equation}
    Then taking the limit $\alpha\downarrow1$ on both sides of the above inequality we obtain (by (\ref{Hcap}))
    \begin{equation}
        \limsup_{n \to \infty} {1 \over n} \log\log N(n,\lambda_1,\lambda_2) \leq \log |A| + C(\N),
    \end{equation}
    where $C(\N)$ is the classical capacity.
\end{proof}

\begin{corollary}\label{cor: quantum covering bound for depolarizing}
    For the qubit depolarizing channel, the classical identification capacity satisfies the following strong converse bound:
\begin{equation}
    C_{\mathrm{ID}}(\mathcal{N}_p) \le 2-h(p/2).
\end{equation}
\end{corollary}

\begin{proof}
    Use \zcref{thm: general covering bound} and the fact that the classical capacity for the qubit depolarizing channel is given by~\cite{king2003capacity}: $C(\mathcal{N}_p)=1-h(p/2)$.

\end{proof}

%%%%%%%%%%%%%%%%%%%%%%%%%%%%%%%%%%%%%%%%%%%%%%%%%%%%%%%%%%%%%%%%%%%%%%%%%%%%%%%%
%%%%%%%%%%%%%%%%%%%%%%%%%%%%%%%%%%%%%%%%%%%%%%%%%%%%%%%%%%%%%%%%%%%%%%%%%%%%%%%%

\section{Summary and Open Questions}
\label{sec: open question}

 In this paper, we derive strong converse identification bounds for the qubit depolarizing channel with noise parameter $p$. These bounds reduce to zero in the limit $p \to 1$, thereby correctly reflecting the fact that reliable identification of classical messages is impossible
over a completely noisy channel. In addition, for simultaneous classical identification under the restriction of complete product measurements, our converse bound matches the corresponding achievability result, hence establishing that the identification capacity in this setting coincides with the classical capacity  of the channel. Note that establishing converse bounds on the identification capacity of fully quantum channels has been a challenging open problem. The only previously known result, due to Atif, Pradhan, and Winter~\cite{atif2024quantum}, has the disadvantage of being strictly positive even in the case of a completely noisy channel. We expect our results to extend straightforwardly to the case of a qudit depolarizing channel.

The most important question stemming from our work is the following: {\em{Can one evaluate the simultaneous identification capacity of the qubit depolarizing channel without imposing the restriction of complete product measurements? Is the capacity in this case also equal to the classical capacity of the channel? In particular, can entangled measurements improve identification rates?}} In the case of message transmission, King~\cite{king2003capacity} showed that the capacity of the depolarizing channel is achievable using product-state inputs and complete product measurements. Whether an analogous statement holds for identification remains an important open question. One difficulty is that, unlike transmission—where the Holevo–Schumacher–Westmoreland theorem provides a tractable characterization of the capacity in terms of entropic quantities—no closed-form expression is currently available for identification over fully quantum channels.

Our proof for the case of simultaneous identification (\zcref{thm: sim depolarizing strong converse}) relies crucially on the fact that, under complete product measurements, the resulting output distributions coincide with those generated by a classical binary symmetric channel. For general entangled measurements, however, such a reduction appears unlikely, and new techniques seem necessary. In fact, in the case of two uses of the depolarizing channel $\mathcal{N}^{\otimes 2}_p$, it can be shown that the Bell measurement does not improve identification rates, as its output distributions are strictly contained within those achievable by 2-qubit complete product measurements. However, there exist partially and asymmetrically entangled measurements whose output distributions lie outside this set. Thus, even in this simplest nontrivial setting, it remains unclear whether entanglement can enhance identification rates.

\bigskip

\textbf{Acknowledgments.} --- We gratefully acknowledge Pau Colomer and Johannes Rosenberger for extensive and invaluable discussions on identification capacities of classical and quantum channels. BB and ND are supported by the Engineering and Physical Sciences Research Council [Grant Ref: EP/Y028732/1]. LY is also supported by the Engineering and Physical Sciences Research Council.

%%%%%%%%%%%%%%%%%%%%%%%%%%%%%%%%%%%%%%%%%%%%%%%%%%%%%%%%%%%%%%%%%%%%%%%%%%%%%%%%
%%%%%%%%%%%%%%%%%%%%%%%%%%%%%%%%%%%%%%%%%%%%%%%%%%%%%%%%%%%%%%%%%%%%%%%%%%%%%%%%
%%%%%%%%%%%%%%%%%%%%%%%%%%%%%%%%%%%%%%%%%%%%%%%%%%%%%%%%%%%%%%%%%%%%%%%%%%%%%%%%

 \bibliographystyle{apsca}
 \bibliography{main}

\newcommand{\etalchar}[1]{$^{#1}$}
\begin{thebibliography}{MLDS{\etalchar{+}}13}
\makeatletter
\providecommand \@ifxundefined [1]{%
 \@ifx{#1\undefined}
}%
\providecommand \@ifnum [1]{%
 \ifnum #1\expandafter \@firstoftwo
 \else \expandafter \@secondoftwo
 \fi
}%
\providecommand \@ifx [1]{%
 \ifx #1\expandafter \@firstoftwo
 \else \expandafter \@secondoftwo
 \fi
}%
\providecommand \natexlab [1]{#1}%
\providecommand \emph  [1]{``#1''}%
\providecommand \bibnamefont  [1]{#1}%
\providecommand \bibfnamefont [1]{#1}%
\providecommand \citenamefont [1]{#1}%
\providecommand \href@noop [0]{\@secondoftwo}%
\providecommand \href [0]{\begingroup \@sanitize@url \@href}%
\providecommand \@href[1]{\@@startlink{#1}\@@href}%
\providecommand \@@href[1]{\endgroup#1\@@endlink}%
\providecommand \@sanitize@url [0]{\catcode `\\12\catcode `\$12\catcode `\&12\catcode `\#12\catcode `\^12\catcode `\_12\catcode `\%12\relax}%
\providecommand \@@startlink[1]{}%
\providecommand \@@endlink[0]{}%
\providecommand \url  [0]{\begingroup\@sanitize@url \@url }%
\providecommand \@url [1]{\endgroup\@href {#1}{\urlprefix }}%
\providecommand \urlprefix  [0]{URL }%
\providecommand \Eprint [0]{\href }%
\providecommand \doibase [0]{http://dx.doi.org/}%
\providecommand \selectlanguage [0]{\@gobble}%
\providecommand \bibinfo  [0]{\@secondoftwo}%
\providecommand \bibfield  [0]{\@secondoftwo}%
\providecommand \translation [1]{[#1]}%
\providecommand \BibitemOpen [0]{}%
\providecommand \bibitemStop [0]{}%
\providecommand \bibitemNoStop [0]{.\EOS\space}%
\providecommand \EOS [0]{\spacefactor3000\relax}%
\providecommand \BibitemShut  [1]{\csname bibitem#1\endcsname}%
\let\auto@bib@innerbib\@empty
%</preamble>
\bibitem[AD02]{ahlswede2002identification}
\bibfield  {author} {\bibinfo {author} {\bibfnamefont {R.}~\bibnamefont {Ahlswede}}\ and\ \bibinfo {author} {\bibfnamefont {G.}~\bibnamefont {Dueck}},\ }\bibfield  {title} {\emph {\bibinfo {title} {Identification via channels},}\ }\href@noop {} {\bibfield  {journal} {\bibinfo  {journal} {IEEE Transactions on Information Theory}\ }\textbf {\bibinfo {volume} {35}},\ \bibinfo {pages} {15} (\bibinfo {year} {2002})}\BibitemShut {NoStop}%
\bibitem[APW24]{atif2024quantum}
\bibfield  {author} {\bibinfo {author} {\bibfnamefont {T.~A.}\ \bibnamefont {Atif}}, \bibinfo {author} {\bibfnamefont {S.~S.}\ \bibnamefont {Pradhan}}, \ and\ \bibinfo {author} {\bibfnamefont {A.}~\bibnamefont {Winter}},\ }\bibfield  {title} {\emph {\bibinfo {title} {Quantum soft-covering lemma with applications to rate-distortion coding, resolvability and identification via quantum channels},}\ }\href@noop {} {\bibfield  {journal} {\bibinfo  {journal} {International Journal of Quantum Information}\ }\textbf {\bibinfo {volume} {22}},\ \bibinfo {pages} {2440013} (\bibinfo {year} {2024})}\BibitemShut {NoStop}%
\bibitem[AW02]{ahlswede2002strong}
\bibfield  {author} {\bibinfo {author} {\bibfnamefont {R.}~\bibnamefont {Ahlswede}}\ and\ \bibinfo {author} {\bibfnamefont {A.}~\bibnamefont {Winter}},\ }\bibfield  {title} {\emph {\bibinfo {title} {Strong converse for identification via quantum channels},}\ }\href@noop {} {\bibfield  {journal} {\bibinfo  {journal} {IEEE Transactions on Information Theory}\ }\textbf {\bibinfo {volume} {48}},\ \bibinfo {pages} {569} (\bibinfo {year} {2002})}\BibitemShut {NoStop}%
\bibitem[BHL{\etalchar{+}}05]{bennett2005remote}
\bibfield  {author} {\bibinfo {author} {\bibfnamefont {C.~H.}\ \bibnamefont {Bennett}}, \bibinfo {author} {\bibfnamefont {P.}~\bibnamefont {Hayden}}, \bibinfo {author} {\bibfnamefont {D.~W.}\ \bibnamefont {Leung}}, \bibinfo {author} {\bibfnamefont {P.~W.}\ \bibnamefont {Shor}}, \ and\ \bibinfo {author} {\bibfnamefont {A.}~\bibnamefont {Winter}},\ }\bibfield  {title} {\emph {\bibinfo {title} {Remote preparation of quantum states},}\ }\href@noop {} {\bibfield  {journal} {\bibinfo  {journal} {IEEE Transactions on Information Theory}\ }\textbf {\bibinfo {volume} {51}},\ \bibinfo {pages} {56} (\bibinfo {year} {2005})}\BibitemShut {NoStop}%
\bibitem[BL17]{bracher2017identification}
\bibfield  {author} {\bibinfo {author} {\bibfnamefont {A.}~\bibnamefont {Bracher}}\ and\ \bibinfo {author} {\bibfnamefont {A.}~\bibnamefont {Lapidoth}},\ }\bibfield  {title} {\emph {\bibinfo {title} {Identification via the broadcast channel},}\ }\href@noop {} {\bibfield  {journal} {\bibinfo  {journal} {IEEE Transactions on Information Theory}\ }\textbf {\bibinfo {volume} {63}},\ \bibinfo {pages} {3480} (\bibinfo {year} {2017})}\BibitemShut {NoStop}%
\bibitem[CDBW25]{colomer2025zero}
\bibfield  {author} {\bibinfo {author} {\bibfnamefont {P.}~\bibnamefont {Colomer}}, \bibinfo {author} {\bibfnamefont {C.}~\bibnamefont {Deppe}}, \bibinfo {author} {\bibfnamefont {H.}~\bibnamefont {Boche}}, \ and\ \bibinfo {author} {\bibfnamefont {A.}~\bibnamefont {Winter}},\ }\bibfield  {title} {\emph {\bibinfo {title} {Zero-entropy encoders and simultaneous decoders in identification via quantum channels},}\ }in\ \href@noop {} {\emph {\bibinfo {booktitle} {Information Theory and Related Fields: Festschrift in Memory of Ning Cai}}}\ (\bibinfo  {publisher} {Springer},\ \bibinfo {year} {2025})\ pp.\ \bibinfo {pages} {478--502}\BibitemShut {NoStop}%
\bibitem[CG23]{cheng2023error}
\bibfield  {author} {\bibinfo {author} {\bibfnamefont {H.-C.}\ \bibnamefont {Cheng}}\ and\ \bibinfo {author} {\bibfnamefont {L.}~\bibnamefont {Gao}},\ }\bibfield  {title} {\emph {\bibinfo {title} {Error exponent and strong converse for quantum soft covering},}\ }\href@noop {} {\bibfield  {journal} {\bibinfo  {journal} {IEEE Transactions on Information Theory}\ }\textbf {\bibinfo {volume} {70}},\ \bibinfo {pages} {3499} (\bibinfo {year} {2023})}\BibitemShut {NoStop}%
\bibitem[Dum06]{dumer2006covering}
\bibfield  {author} {\bibinfo {author} {\bibfnamefont {I.}~\bibnamefont {Dumer}},\ }\bibfield  {title} {\emph {\bibinfo {title} {Covering an ellipsoid with equal balls},}\ }\href@noop {} {\bibfield  {journal} {\bibinfo  {journal} {Journal of Combinatorial Theory, Series A}\ }\textbf {\bibinfo {volume} {113}},\ \bibinfo {pages} {1667} (\bibinfo {year} {2006})}\BibitemShut {NoStop}%
\bibitem[HCG25]{hayashi2025resolvability}
\bibfield  {author} {\bibinfo {author} {\bibfnamefont {M.}~\bibnamefont {Hayashi}}, \bibinfo {author} {\bibfnamefont {H.-C.}\ \bibnamefont {Cheng}}, \ and\ \bibinfo {author} {\bibfnamefont {L.}~\bibnamefont {Gao}},\ }\bibfield  {title} {\emph {\bibinfo {title} {Resolvability of classical-quantum channels},}\ }\href@noop {} {\bibfield  {journal} {\bibinfo  {journal} {IEEE Transactions on Information Theory}\ } (\bibinfo {year} {2025})}\BibitemShut {NoStop}%
\bibitem[Hol02]{holevo2002capacity}
\bibfield  {author} {\bibinfo {author} {\bibfnamefont {A.~S.}\ \bibnamefont {Holevo}},\ }\bibfield  {title} {\emph {\bibinfo {title} {The capacity of the quantum channel with general signal states},}\ }\href@noop {} {\bibfield  {journal} {\bibinfo  {journal} {IEEE Transactions on Information Theory}\ }\textbf {\bibinfo {volume} {44}},\ \bibinfo {pages} {269} (\bibinfo {year} {2002})}\BibitemShut {NoStop}%
\bibitem[HV93]{verdhapproximation}
\bibfield  {author} {\bibinfo {author} {\bibfnamefont {T.}~\bibnamefont {Han}}\ and\ \bibinfo {author} {\bibfnamefont {S.}~\bibnamefont {Verdu}},\ }\bibfield  {title} {\emph {\bibinfo {title} {Approximation theory of output statistics},}\ }\href {http://dx.doi.org/10.1109/18.256486} {\bibfield  {journal} {\bibinfo  {journal} {IEEE Transactions on Information Theory}\ }\textbf {\bibinfo {volume} {39}},\ \bibinfo {pages} {752} (\bibinfo {year} {1993})}\BibitemShut {NoStop}%
\bibitem[HW12]{hayden2012weak}
\bibfield  {author} {\bibinfo {author} {\bibfnamefont {P.}~\bibnamefont {Hayden}}\ and\ \bibinfo {author} {\bibfnamefont {A.}~\bibnamefont {Winter}},\ }\bibfield  {title} {\emph {\bibinfo {title} {Weak decoupling duality and quantum identification},}\ }\href@noop {} {\bibfield  {journal} {\bibinfo  {journal} {IEEE Transactions on Information Theory}\ }\textbf {\bibinfo {volume} {58}},\ \bibinfo {pages} {4914} (\bibinfo {year} {2012})}\BibitemShut {NoStop}%
\bibitem[Kin03]{king2003capacity}
\bibfield  {author} {\bibinfo {author} {\bibfnamefont {C.}~\bibnamefont {King}},\ }\bibfield  {title} {\emph {\bibinfo {title} {The capacity of the quantum depolarizing channel},}\ }\href@noop {} {\bibfield  {journal} {\bibinfo  {journal} {IEEE Transactions on Information Theory}\ }\textbf {\bibinfo {volume} {49}},\ \bibinfo {pages} {221} (\bibinfo {year} {2003})}\BibitemShut {NoStop}%
\bibitem[L{\"o}b99]{lober1999quantum}
\bibfield  {author} {\bibinfo {author} {\bibfnamefont {P.}~\bibnamefont {L{\"o}ber}},\ }\bibfield  {title} {\emph {\bibinfo {title} {Quantum channels and simultaneous id coding},}\ }\href@noop {} {\bibfield  {journal} {\bibinfo  {journal} {arXiv preprint quant-ph/9907019}\ } (\bibinfo {year} {1999})}\BibitemShut {NoStop}%
\bibitem[MLDS{\etalchar{+}}13]{muller2013quantum}
\bibfield  {author} {\bibinfo {author} {\bibfnamefont {M.}~\bibnamefont {M{\"u}ller-Lennert}}, \bibinfo {author} {\bibfnamefont {F.}~\bibnamefont {Dupuis}}, \bibinfo {author} {\bibfnamefont {O.}~\bibnamefont {Szehr}}, \bibinfo {author} {\bibfnamefont {S.}~\bibnamefont {Fehr}}, \ and\ \bibinfo {author} {\bibfnamefont {M.}~\bibnamefont {Tomamichel}},\ }\bibfield  {title} {\emph {\bibinfo {title} {On quantum r{\'e}nyi entropies: A new generalization and some properties},}\ }\href@noop {} {\bibfield  {journal} {\bibinfo  {journal} {Journal of Mathematical Physics}\ }\textbf {\bibinfo {volume} {54}} (\bibinfo {year} {2013})}\BibitemShut {NoStop}%
\bibitem[Pet86]{petz1986quasi}
\bibfield  {author} {\bibinfo {author} {\bibfnamefont {D.}~\bibnamefont {Petz}},\ }\bibfield  {title} {\emph {\bibinfo {title} {Quasi-entropies for finite quantum systems},}\ }\href@noop {} {\bibfield  {journal} {\bibinfo  {journal} {Reports on mathematical physics}\ }\textbf {\bibinfo {volume} {23}},\ \bibinfo {pages} {57} (\bibinfo {year} {1986})}\BibitemShut {NoStop}%
\bibitem[RDP23]{rosenberger2023identification}
\bibfield  {author} {\bibinfo {author} {\bibfnamefont {J.}~\bibnamefont {Rosenberger}}, \bibinfo {author} {\bibfnamefont {C.}~\bibnamefont {Deppe}}, \ and\ \bibinfo {author} {\bibfnamefont {U.}~\bibnamefont {Pereg}},\ }\bibfield  {title} {\emph {\bibinfo {title} {Identification over quantum broadcast channels},}\ }\href@noop {} {\bibfield  {journal} {\bibinfo  {journal} {Quantum Information Processing}\ }\textbf {\bibinfo {volume} {22}},\ \bibinfo {pages} {361} (\bibinfo {year} {2023})}\BibitemShut {NoStop}%
\bibitem[R{\'e}n61]{renyi1961measures}
\bibfield  {author} {\bibinfo {author} {\bibfnamefont {A.}~\bibnamefont {R{\'e}nyi}},\ }\bibfield  {title} {\emph {\bibinfo {title} {On measures of entropy and information},}\ }in\ \href@noop {} {\emph {\bibinfo {booktitle} {Proceedings of the fourth Berkeley symposium on mathematical statistics and probability, volume 1: contributions to the theory of statistics}}},\ Vol.~\bibinfo {volume} {4}\ (\bibinfo {organization} {University of California Press},\ \bibinfo {year} {1961})\ pp.\ \bibinfo {pages} {547--562}\BibitemShut {NoStop}%
\bibitem[Sha48]{shannon1948mathematical}
\bibfield  {author} {\bibinfo {author} {\bibfnamefont {C.~E.}\ \bibnamefont {Shannon}},\ }\bibfield  {title} {\emph {\bibinfo {title} {A mathematical theory of communication},}\ }\href@noop {} {\bibfield  {journal} {\bibinfo  {journal} {The Bell system technical journal}\ }\textbf {\bibinfo {volume} {27}},\ \bibinfo {pages} {379} (\bibinfo {year} {1948})}\BibitemShut {NoStop}%
\bibitem[SW97]{schumacher1997sending}
\bibfield  {author} {\bibinfo {author} {\bibfnamefont {B.}~\bibnamefont {Schumacher}}\ and\ \bibinfo {author} {\bibfnamefont {M.~D.}\ \bibnamefont {Westmoreland}},\ }\bibfield  {title} {\emph {\bibinfo {title} {Sending classical information via noisy quantum channels},}\ }\href@noop {} {\bibfield  {journal} {\bibinfo  {journal} {Physical Review A}\ }\textbf {\bibinfo {volume} {56}},\ \bibinfo {pages} {131} (\bibinfo {year} {1997})}\BibitemShut {NoStop}%
\bibitem[Ume62]{umegaki1962conditional}
\bibfield  {author} {\bibinfo {author} {\bibfnamefont {H.}~\bibnamefont {Umegaki}},\ }\bibfield  {title} {\emph {\bibinfo {title} {Conditional expectation in an operator algebra, iv (entropy and information)},}\ }in\ \href@noop {} {\emph {\bibinfo {booktitle} {Kodai Mathematical Seminar Reports}}},\ Vol.~\bibinfo {volume} {14}\ (\bibinfo {organization} {Department of Mathematics, Tokyo Institute of Technology},\ \bibinfo {year} {1962})\ pp.\ \bibinfo {pages} {59--85}\BibitemShut {NoStop}%
\bibitem[Ver15]{verdu2015alpha}
\bibfield  {author} {\bibinfo {author} {\bibfnamefont {S.}~\bibnamefont {Verd{\'u}}},\ }\bibfield  {title} {\emph {\bibinfo {title} {$\alpha$-mutual information},}\ }in\ \href@noop {} {\emph {\bibinfo {booktitle} {2015 Information Theory and Applications Workshop (ITA)}}}\ (\bibinfo {organization} {IEEE},\ \bibinfo {year} {2015})\ pp.\ \bibinfo {pages} {1--6}\BibitemShut {NoStop}%
\bibitem[Win04]{winter2004quantum}
\bibfield  {author} {\bibinfo {author} {\bibfnamefont {A.}~\bibnamefont {Winter}},\ }\bibfield  {title} {\emph {\bibinfo {title} {Quantum and classical message identification via quantum channels},}\ }\href@noop {} {\bibfield  {journal} {\bibinfo  {journal} {arXiv preprint quant-ph/0401060}\ } (\bibinfo {year} {2004})}\BibitemShut {NoStop}%
\bibitem[Win06]{winter2006identification}
\bibfield  {author} {\bibinfo {author} {\bibfnamefont {A.}~\bibnamefont {Winter}},\ }\bibfield  {title} {\emph {\bibinfo {title} {Identification via quantum channels in the presence of prior correlation and feedback},}\ }in\ \href@noop {} {\emph {\bibinfo {booktitle} {General Theory of Information Transfer and Combinatorics}}}\ (\bibinfo  {publisher} {Springer},\ \bibinfo {year} {2006})\ pp.\ \bibinfo {pages} {486--504}\BibitemShut {NoStop}%
\bibitem[Win13]{winter2013identification}
\bibfield  {author} {\bibinfo {author} {\bibfnamefont {A.}~\bibnamefont {Winter}},\ }\bibfield  {title} {\emph {\bibinfo {title} {Identification via quantum channels},}\ }in\ \href@noop {} {\emph {\bibinfo {booktitle} {Information Theory, Combinatorics, and Search Theory: In Memory of Rudolf Ahlswede}}}\ (\bibinfo  {publisher} {Springer},\ \bibinfo {year} {2013})\ pp.\ \bibinfo {pages} {217--233}\BibitemShut {NoStop}%
\bibitem[WWY14]{wilde2014strong}
\bibfield  {author} {\bibinfo {author} {\bibfnamefont {M.~M.}\ \bibnamefont {Wilde}}, \bibinfo {author} {\bibfnamefont {A.}~\bibnamefont {Winter}}, \ and\ \bibinfo {author} {\bibfnamefont {D.}~\bibnamefont {Yang}},\ }\bibfield  {title} {\emph {\bibinfo {title} {Strong converse for the classical capacity of entanglement-breaking and hadamard channels via a sandwiched r{\'e}nyi relative entropy},}\ }\href@noop {} {\bibfield  {journal} {\bibinfo  {journal} {Communications in Mathematical Physics}\ }\textbf {\bibinfo {volume} {331}},\ \bibinfo {pages} {593} (\bibinfo {year} {2014})}\BibitemShut {NoStop}%
\end{thebibliography}%

%%%%%%%%%%%%%%%%%%%%%%%%%%%%%%%%%%%%%%%%%%%%%%%%%%%%%%%%%%%%%%%%%%%%%%%%%%%%%%%%
%%%%%%%%%%%%%%%%%%%%%%%%%%%%%%%%%%%%%%%%%%%%%%%%%%%%%%%%%%%%%%%%%%%%%%%%%%%%%%%%

%%%%%%%%%%%%%%%%%%%%%%%%%%%%%%%%%%%%%%%%%%%%%%%%%%%%%%%%%%%%%%%%%%%%%%%%%%%%%%%%
%%%%%%%%%%%%%%%%%%%%%%%%%%%%%%%%%%%%%%%%%%%%%%%%%%%%%%%%%%%%%%%%%%%%%%%%%%%%%%%%

\end{document}